%% file: main.tex
\documentclass[acmtog,nonacm]{acmart}

\setcopyright{none}
\usepackage{booktabs} % For formal tables
\usepackage{enumitem}

\usepackage{algorithm} % Algorithm float environment
\usepackage{algpseudocode} % Algorithmic pseudocode (\State, \Require, etc.)
\floatname{algorithm}{ALGORITHM}
\usepackage{bm} % Algorithm float environment

\acmJournal{TOG}

\newif\ifshowmemos
\newcommand{\s}{{\bm{s}}}
\newcommand{\brho}{{\bm{\rho}}}
\newcommand{\bX}{{\bm{X}}}
\newcommand{\bx}{{\bm{x}}}
\newcommand{\bY}{{\bm{Y}}}

\newcommand{\bu}{{\bm{u}}}
\newcommand{\bz}{{\bm{z}}}
\newcommand{\bZ}{{\bm{Z}}}
\newcommand{\bd}{{\bm{d}}}
\newcommand{\bmu}{{\bm{\mu}}}
\newcommand{\bell}{{\bm{\ell}}}
\newcommand{\beps}{{\bm{\epsilon}}}

\begin{document}
% Title portion
%\title{The Stochasticity-Determinism Tradeoff in Single-Photon Sensing}
\title{One Photon, Many Worlds: Posteriors and Predictions with Single-Photon Cameras}
\author{Haejoon Lee}
\affiliation{%
  \institution{Carnegie Mellon University}
  \city{Pittsburgh}
  \state{PA}
  \country{USA}}
\email{haejoonl@andrew.cmu.edu}

\author{Mohit Gupta}
\affiliation{%
  \institution{University of Wisconsin-Madison}
  \city{Madison}
  \state{WI}
  \country{USA}}
\email{mohitg@cs.wisc.edu}

\author{Vijayakumar Bhagavatula}
\affiliation{%
  \institution{Carnegie Mellon University}
  \city{Pittsburgh}
  \state{PA}
  \country{USA}}
\email{vk16@andrew.cmu.edu}

\author{Aswin C. Sankaranarayanan}
\affiliation{%
  \institution{Carnegie Mellon University}
  \city{Pittsburgh}
  \state{PA}
  \country{USA}}
\email{saswin@andrew.cmu.edu}

\begin{teaserfigure}
\centering
  \includegraphics[width=0.89\textwidth]{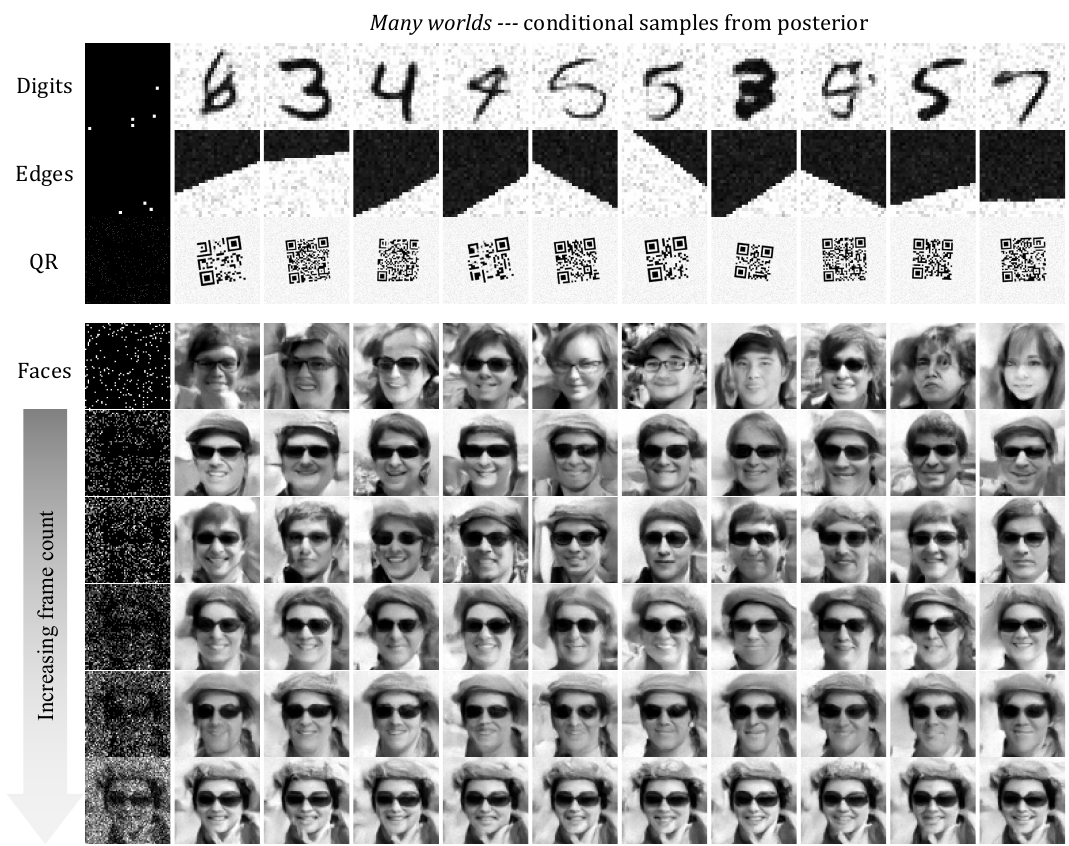}
  \caption{SPAD cameras operated in binary-frame mode record at most one detection per pixel per exposure. The top four rows shows samples from the conditional posterior given a single (simulated) binary SPAD frame, shown in the first column where each pixel has measured, at most,\textit{one photon}. In such photon starved regimes, the inverse problem of image reconstruction is one-to-many, leading to \textit{many worlds} of possibilities in the output. This uncertainty reduces as we increase the number of binary SPAD frames, as seen in subsequent rows. The conditional samples were produced by training diffusion models for each of the four applications (digits, edges, QR codes, and faces), and empirically characterizes how uncertainty in photon arrivals map to uncertainty in reconstruction and inference. In the case of the faces data, notice how different features like sunglasses, gender, facial expression, and identity take different number of frames to resolve correctly.}
  \label{fig:teaser}
\end{teaserfigure}

\begin{abstract}
  Single-photon avalanche diode (SPAD) cameras operate fundamentally differently from conventional cameras due to their photon-counting nature. 
Each frame produces a binary image: pixels report zero if no photons arrived during exposure, and one if one or more photons arrived.
Reconstructing a scene or inferring its properties from a single binary frame is difficult because many different images could produce the same measurement; thus, the inverse problem is fundamentally one-to-many.
As we gather more binary measurements, the inherent uncertainty associated with the inverse problem and any associated inference diminishes.
With sufficient photon counts, photon noise becomes negligible relative to the signal mean, enabling near-deterministic scene recovery and inference.
This work characterizes the transition from stochastic to near-deterministic scene understanding as photon budget increases, analyzing how the stochasticity in photon arrival affects downstream inference tasks.
Technically, we develop a conditional generative framework based on a Hypergeometric frame-thinning process for accumulated binary SPAD measurements.
Generative models capture the one-to-many nature of photon-starved inverse problems, enabling empirical characterization of how this ambiguity diminishes with increasing measurements and its impact on downstream tasks like character recognition, QR code decoding, and facial analysis.
\end{abstract}




\maketitle

\section{Introduction}
\label{sec:intro}
\input{intro.tex}

\section{Prior work}
\label{sec:prior}
\input{prior.tex}

\section{Overview}
\label{sec:overview}
\input{overview.tex}
\section{Conditional Image Generation}
\label{sec:method}

\input{method.tex}

\input{cases.tex}

\section{Discussions}
\label{sec:discuss}

\input{discuss.tex}

%%
%% The next two lines define the bibliography style to be used, and the
%% bibliography file.
\bibliographystyle{ACM-Reference-Format}
\nocite{*}
\bibliography{refs}

\clearpage
\appendix
% \section{Appendix}
\label{sec:appendix}
\input{appendix.tex}

\end{document}

%% file: intro.tex
%In the beginning, there was \textit{darkness}. 
%%
%Prior to exposure, no photons had yet reached the SPAD sensor, and the scene before the camera held a multitude of possibilities.
%%
%Then, the first measurements were made, with some pixels registering photon arrivals and sketching a binary image. 
%%
%Though the possibilities narrowed, ample uncertainty persisted in what lay before the camera.
%%
%As successive binary frames  accumulated, this uncertainty in the scene dramatically reduced.
%%
%%At some point, this stochasticity in the nature of the scene was largely eliminated and we had a deterministic understanding of the scene and its properties.
%%
%Lo and behold, stochasticity  yielded to determinism; with enough \textit{light},  the scene and its properties were resolved into stark clarity.
%%
%This work examines this arc from stochasticity to determinism, tracing how randomness in photon arrivals, through the SPAD's photon-counting nature, impacts downstream inference.

Imagine a photon counting sensor, capable of resolving light at the granularity of individual photons.
Prior to exposure, this camera has no knowledge of the scene in front of it, and so there are multitude of possibilities on the image it could acquire.
Once the first measurements are made, some pixels register photon arrivals, we have a binary sketch of the scene.
This narrows the space of images that we could be imaging, albeit not by much.
As the sensor gathers more measurements, the uncertainty in the image being observed decreases progressively.
At some point, we enter the realm of well exposed photographs, where uncertainty is reduced so significantly that the acquired measurements convey a sense of determinism (although the uncertainty never quite goes to zero).
%Although the possibilities  are narrowed, there is still ample uncertainty in the scene being imaged.
%
%As successive binary frames accumulate, this uncertainty in the scene is progressively reduced.
%
%At some point, this stochasticity in the nature of the scene is largely eliminated and we have a near-deterministic understanding of the image and its properties
%%
%Lo and behold, stochasticity  yields to determinism; with enough light,  the scene and its properties were resolved into stark clarity.
%%
This work examines this arc of how uncertainty in scene reconstruction and downstream inference changes as the photon budget increases; effectively, we seek to characterize how randomness in photon arrivals, through the  photon-counting nature of single-photon cameras, manifests itself in reconstruction and inference tasks.

The central thesis of this work is how the photon budget determines the ability to interpret a scene and its properties, in the context of a single-photon avalanche diodes (SPADs), a specific kind of single-photon image sensor.
Consider the scenario where a SPAD sensor is observing a (static) scene, making a sequence of  measurements, where each measurement is a binary frame indicating whether photons arrived at each pixel during the exposure time.
The stochasticity of such binary measurements is well understood.
The measurement at a pixel in each binary frame is a Bernoulli random variable; at low flux, the mean of this random variable is equal to the incident light level.
When the SPAD acquires multiple such binary measurements, their sum at each pixel is a binomial distribution which, at large measurement counts, is well approximated by a Poisson distribution.
Both the mean and the variance of this Poisson distribution increase linearly with the number of measurements made by the sensor.
This is the regime of a well-exposed photograph; although photon noise dominates,  the signal-to-noise ratio of intensity at each pixel monotonically increases with accumulating SPAD measurements.
Hence, the output of any deterministic operator on the image—say, a classifier or a regressor—becomes largely independent of the photon noise. %; we call this the \textit{deterministic regime}.

At low photon budgets, we encounter an entirely different problem.
As an illustrative example, consider the extreme scenario of a single binary measurement from the SPAD sensor.
As before, we have a precise characterization of the stochasticity of the measurement in terms of its probability mass function and its dependence on the incident flux  at each sensor pixel.
However, if we take an arbitrary natural image as the candidate for the scene being observed, then the posterior probability of this image given the measurement is likely similar to that of any other image.
This reveals that the mapping from a single binary measurement to the space of images is fundamentally \textit{one-to-many}.
The problem persists as we gradually increase the number of measurements, with the one-to-many nature of the problem gradually collapsing to one-to-one  in the limit when the signal-to-noise ratio goes to zero.

\textit{We are interested in characterizing the one-to-many nature of the image reconstruction problem underlying SPAD measurements at low photon budgets.}
Our primary insight is that generative models are uniquely suited for studying this problem as they can learn functional mappings between two random variables.
This naturally aligns with our setting since we want to characterize how randomness in SPAD measurements maps to uncertainty in image space.
By leveraging the generative aspects of these models, we can empirically sample the space of images that are tied to a single or multiple SPAD measurements.
This allows us to characterize the nature of the one-to-many mapping problem, and  how  uncertainty in the image space evolves with the number of measurements.
Consequently, we can also characterize the performance of classifiers and regressors acting on the measurements as uncertainty in the measurements and reconstructions are both reduced.

\paragraph{Contributions.} This paper provides modeling tools and a range of empirical results that characterize uncertainty in reconstruction and inference with increasing photon budget. 

\begin{itemize}[leftmargin=*]
\item \textit{A model for frame thinning in SPAD measurements.} 
We show that frame thinning—the process of discarding binary measurements—can be modeled as a Markov chain with hypergeometric transitions.
This allows us to convert  large frame count measurements, progressively, into those made with small frame counts, in a gradual process that lends itself well to stable inversion.
\item \textit{Inverse Markov models for SPAD restoration.} We design a reverse Markov process that inverts frame thinning while preserving Binomial marginals, implemented as a variational autoencoder-style conditional sampler.

\item \textit{Empirical validation across multiple applications.}  We evaluate the approach on edge regression, optical character recognition, QR code decoding, and facial reconstruction. For each task, we characterize how the one-to-many ambiguity resolves and how reconstruction and inference accuracy improve with increasing photon budget.
\end{itemize}

Figure \ref{fig:teaser} shows example reconstructions for the four applications we consider, and how the one-to-many nature of the inverse problem manifests itself in each instance. 

\paragraph{Impact.} These contributions provide tools to answer fundamental questions about single-photon imaging. While acquiring more photons intuitively improves results, the minimal photon budget required for reliable reconstruction or inference has remained uncharacterized. 
Our diffusion-based approach provides quantitative answers to these questions.
We further characterize the one-to-many ambiguity by providing tools to quantify reconstruction uncertainty across photon budgets.

%% file: prior.tex
\paragraph{Diffusion models.} Probabilistic diffusion models have emerged as a powerful approach to image generation.
The  premise behind them lies in creating a forward diffusion model that progressively converts samples from a structured distribution like natural images into pure random noise; subsequently, an inverse model reverses the effect of the forward model, one step at a time, to recover the structure of the original data \citep{sohl2015deep}.
\citet{ho2020denoising} showed that such diffusion models could be used to generate high quality images through a learned denoising process.
Subsequently, \citet{song2020score} connected diffusion models to solutions of stochastic differential equations.

While Gaussian noise is the most common corrupting noise used in diffusion models, there has been work focused on models built around Poisson and binomial distributions.
\citet{nachmani2021non} and \citet{xie2023diffusion} expand the traditional models to mixture of Guassian, Gamma distributions, and Poisson distributions arguing that this provides improved degrees of freedoms and improved results in generative modeling.
\citet{chen2023learning} used Poisson distribution for the marginal at each time step, defining a thinning process for the forward process, and a photon aggregation process for the reverse process. 
In the context of denoising single photon emission computed tomography (SPECT) sinograms, that have Poisson statistics, \citet{lai2025poisson} avoid Gaussian diffusion models to build a forward process where a normal dose (or high SNR) SPECT measurement is converted to a low dose one; this is done by removing Poisson random variables progressively (a process that is similar to the frame thinning approach we adopt). 
However, as they note in their work, the reverse model does not have a tractable distribution, which leads them to resort to inverse CDF sampling.
In contrast to these approaches, we use a Binomial process that is perfectly tailored to SPAD measurements; for this distribution, we provide a principled approach to both the forward and the reverse process, carefully ensuring that all distributions involved are Binomial.

\paragraph{Denoising techniques with diffusion models}
Since diffusion models provide a way to model a space of signals, they have been widely used to solve inverse problems  \cite{zhu2023denoising}.
\citet{chung2023diffusion} expand diffusion models, conditioned to solve inverse problems, to include Poisson measurement noise.
\citet{melidonis2025score} expand the noise models beyond Poisson to include Binomial distributions, arguing that photon starved measurement require the use of such models.

Closely related to this paper is the use of diffusion models to denoise SPAD images, in order to remove the effects of photon noise.
Quanta diffusion \citep{chennuri2025quanta} study the use of diffusion models for video reconstruction from SPAD data; they adapt a standard Gaussian diffusion model, conditioned with noisy SPAD data, to recover a video.
gQIR \citep{garg2026gqir} fine-tune a pre-trained diffusion backbone and train it adversarially as a one-step deterministic reconstruction model.
Recent work by \citet{dvir2026diffusion} uses the likelihood model underlying the measurements made by SPAD sensors to derive a score function, which is used with a diffusion model to denoise the measurements.

While the prior work discussed above have explored diffusion models for inverse problems and photon-limited imaging, they have largely ignored the one-to-many nature of the inverse problem.
For the most part, they treat the inverse problem as a one-to-one map, where the mean square error (or variants thereof) are used to guide the recovery, under the regularization provided by the diffusion process.
In contrast, our work tackles the one-to-many nature of the inverse problem head on, developing a forward model based on frame thinning, and using a principled approach to learn the reverse process.

%\paragraph{Computational imaging with SPADs} 

%%passive ultrawideband \cite{wei2023passive}

%% file: overview.tex
In this section, we describe the overview of our work.

\begin{figure}[t]
    \centering
    \includegraphics[width=1\columnwidth]{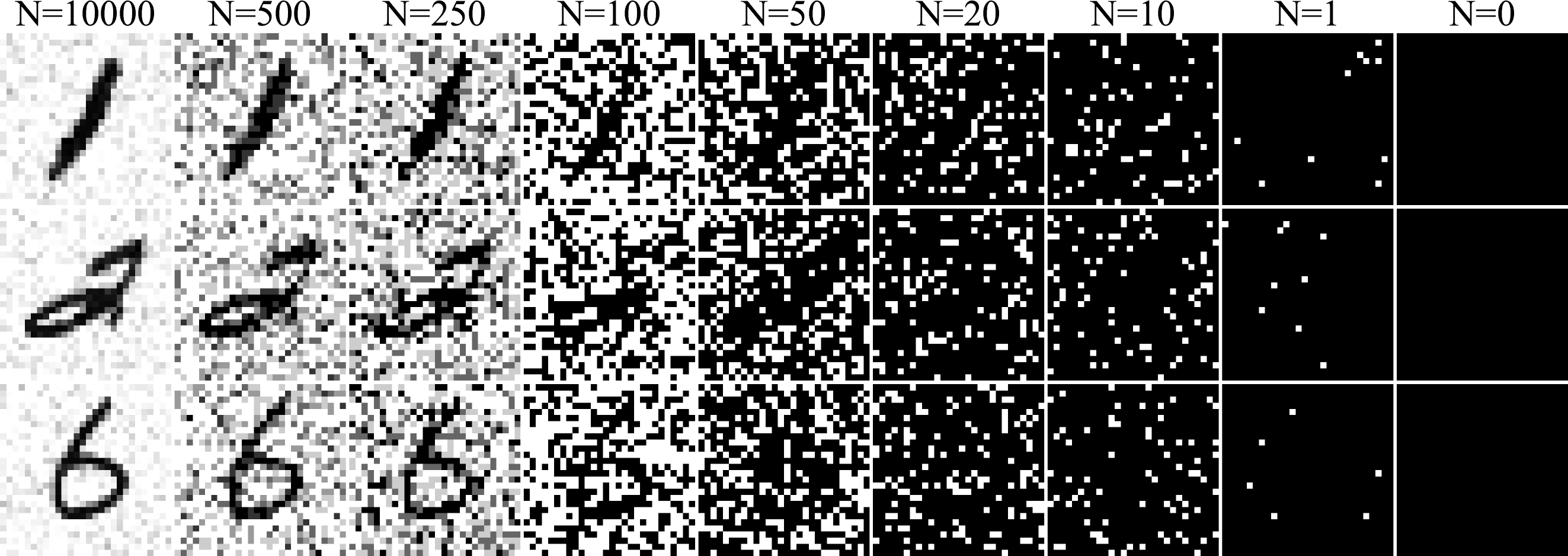}
    \caption{Example MNIST images simulated with the SPAD sensor model. The maximum expected photon level per binary frame is set to $\tau=0.01$.
    }
    \label{fig:mnist_random_spad_grid}
\end{figure}

\subsection{SPAD  Measurement Model}

%We begin  describe how SPAD sensors form an image. In SPAD image sensor arrays, each pixel is operated over a sequence of short binary frames. In each frame, the pixel records whether at least one photon arrives during the gating interval. Thus, each frame produces a binary measurement: \textit{one} if at least one photon is detected and \textit{zero} otherwise.

We begin with a brief description of how SPAD image sensors operate.
Each pixel in a SPAD sensor is a photodiode operated beyond its breakdown voltage; as a result, a single photoelectron triggers an avalanche reaction that saturates the pixel.
In effect, each SPAD pixel registers whether one or more photons arrived during the exposure time.
%
% In effect, each SPAD pixel records a binary event indicating whether at least one photon arrived during the gating interval.
Thus, the SPAD sensor produces a binary image in each exposure, with each pixel reporting zero if no photons arrived and one if one or more photons arrived.

Formalizing this, consider a SPAD sensor with $D$ pixels.
Assuming a static scene that is focused on the sensor, let $\s  =(s_1, \ldots, s_D) \in \mathbb{R}_+^D$  be the mean photon arrival rate for  the sensor pixels, normalized such that $\| \s \|_\infty = 1$, i.e., the brightest scene point is normalized to unit mean arrival rate.
Assuming Poisson-distributed  photon arrivals, the probability of detecting at least one photon at pixel $i$ is 
\begin{equation}
\rho_i = 1 - e^{-\tau s_i},
\qquad
\brho = (\rho_1,\dots,\rho_D) \in (0,1)^D .
\label{eq:rho-def}
\end{equation}
Here, $\tau$ is a scalar that captures the effect of exposure time, quantum efficiency of the sensor, as well as overall scene brightness (since $\s$ is the normalized photon arrival rate).
Thus, the measurement at a pixel $i$ in a single frame is a Bernoulli random variable with success probability $\rho_i$.

Suppose we acquire $N$ such binary frames.
Since  the individual frames are statistically independent, their sum---denoted as $\bX^{(N)}$---is the sufficient statistic for the underlying random variable.
The $i$-th element of $\bX^{(N)}$, denoted as $X^{(N)}_i$, follows a Binomial distribution:
\begin{equation}
X^{(N)}_i \sim \operatorname{Binomial}(N,\rho_i).
\label{eq:spad_pixel_measurement}
\end{equation}
Here, we assume an ideal sensor with negligible dark current.
Figure~\ref{fig:mnist_random_spad_grid} shows MNIST examples simulated using this SPAD measurement model.
The mean and variance of $X^{(N)}_i$ are given by
\begin{equation}
\mathbb{E}[X^{(N)}_i] = N\rho_i,
\qquad
\operatorname{Var}[X^{(N)}_i] = N\rho_i(1-\rho_i).
\end{equation}
Under the definition $\mathrm{SNR}=\mathbb{E}[X]^2/\operatorname{Var}[X]$, this gives
\begin{equation}
\mathrm{SNR}_i
=
\frac{N\rho_i}{1-\rho_i}.
\end{equation}
Therefore, increasing the number of binary frames improves measurement quality. 
As we will see next, the number of measurements we acquire also influences the nature of the inverse problem we will need to solve.
%
%However, if each frame requires exposure time $\Delta t$, then the total acquisition time is $N\Delta t$. 
%
%This creates a fundamental trade-off:  larger $N$ improves the signal-to-noise ratio but reduces temporal resolution and increases motion blur.
%

% Conditioned on the photon-detection probability field $\rho$, $D$ independent SPAD pixel measurements yields a count image
% \begin{equation}
% X^{(N)} \mid \rho
% \;\sim\;
% \prod_{i=1}^D \operatorname{Binomial}\!\big(N,\rho_i\big).
% \label{eq:physical-model}
% \end{equation}
% Thus each pixel count lies in $\{0,1,\dots,N\}$. With an exposure time $\Delta t$ for each binary SPAD frame, the total capture time for single SPAD image is $N \Delta t$.

\subsection{Photon Starvation Makes Restoration One-to-Many}
\label{sec:overview-one-to-many}
We study the problem of estimating the noise-free image formed on the sensor.
From Figure~\ref{fig:mnist_random_spad_grid}, we can immediately observe that this problem is extremely difficult when $N$, the number of binary frames, is small.
In part, this is due to the measurement likelihood being flat over the space of images, i.e., there are a large collection of images which are equally likely to have resulted in the measurements that we observe.
In this regime, the inverse problem  of estimating the noise-free image is intrinsically \textit{one-to-many}. 
For larger $N$, the measurements become inherently more informative and noise-free.
Hence, the conditional distribution,
\begin{equation}
q\!\left(\bX^{(N_1)} \mid \bX^{(N_k)}\right), 
\end{equation}
where $N_1 \gg N_k$, is broad for small $N_k$ and  narrows as $N_k$ increases. 
This narrowing occurs at different rates for different image features.
For example, as seen in Figure~\ref{fig:mnist_random_spad_grid},  digit identity is resolved for small $N$ (say, $N=20$ or $50$), but a high quality reconstruction with pixel level details is only achieved at much larger $N$ (say, $N=1000$ or so).

\begin{figure}[t]
    \centering
    \includegraphics[width=1\columnwidth]{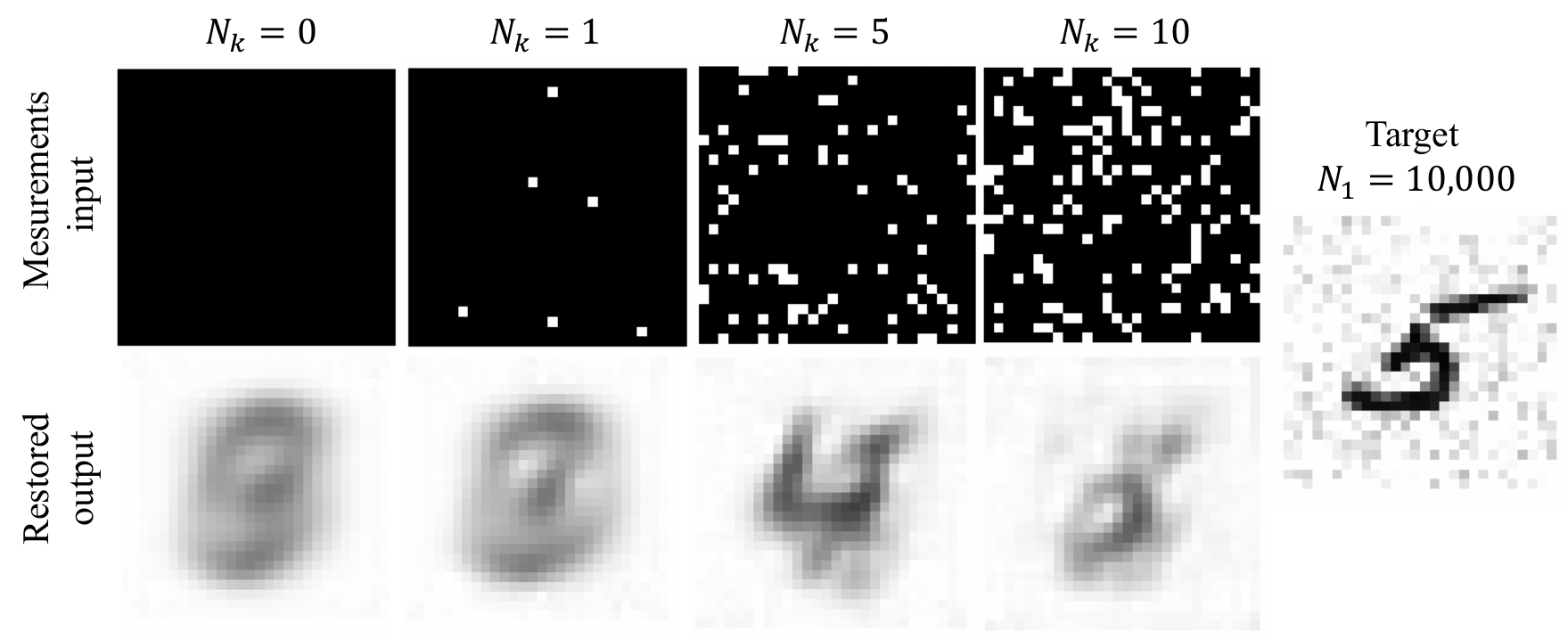}
    \caption{Deterministic restoration under photon starvation. For very small \(N\), the same low-count SPAD observation is compatible with many high-count targets. A deterministic denoiser trained with a L2 reconstruction loss produces a single averaged output, which can be blurry and cannot represent the remaining conditional uncertainty.}
    \label{fig:mnist_determin_recon}
\end{figure}

A deterministic denoiser cannot capture the one-to-many structure endemic to low $N_k$.
%This one-to-many structure is not well represented by a deterministic denoiser.
%
If a deterministic network $f$ is trained with $L_2$ loss to predict $\bX^{(N_1)}$ from $\bX^{(N_k)}$, then the empirical-risk minimizer becomes the conditional mean:
\begin{equation}
f^\star\!\left(\bx_k\right)
=
\mathbb E\!\left[\bX^{(N_1)} \mid \bX^{(N_k)}= \bx_k\right].
\end{equation}

When the conditional distribution is multimodal---comprising of  multiple distinct digits or shapes---the conditional mean is extremely blurred, as illustrated in Figure~\ref{fig:mnist_determin_recon}. 
Since a deterministic network always produces the same output for a given input, it cannot handle multiple plausible reconstructions; instead, the uncertainty only appears as blurriness in the output.

\subsection{Solution Overview}
\label{sec:overview-markovian-conditional-generation}

The failure of deterministic models suggests that we should treat  SPAD restoration, especially in photon-starved regimes,  as  a conditional image generation problem.
Given a fixed low-count measurement $\bX^{(N_k)} = \bx_k$, a conditional generator can produce multiple samples from an approximation to
\begin{equation}
q\!\left(\bX^{(N_1)} \mid \bX^{(N_k)} = \bx_k\right).
\end{equation}
By learning a mapping between the two random variables, we can sample from the  conditional distribution and resolve the one-to-many 
nature of the inverse problem.

%The question we ask in this paper is not only how to reconstruct a clean-looking image from a low-count SPAD measurement, but also how much ambiguity remains after a given number of binary frames. 
%
%When $N_k$ is very small, the observed image \(X^{(N_k)}\) contains little information about the scene, so the conditional distribution of the clean measurement (i.e., $X^{(N_1)} \mid X^{(N_k)}$) is broad and close to unconditional image generation. 
%
%When \(N_k\) is large, the same conditional distribution becomes narrow and close to deterministic reconstruction. 
%
%A direct  approach would be to learn the posterior distribution of the latent scene image with a conditional diffusion model, flow model, or diffusion posterior sampler (e.g., \cite{chung2023diffusion}), and then sample the clean SPAD image with large $N_1$ through the Binomial measurement model. 
%

\begin{figure*}[t]
    \centering
    \includegraphics[width=0.9\textwidth]{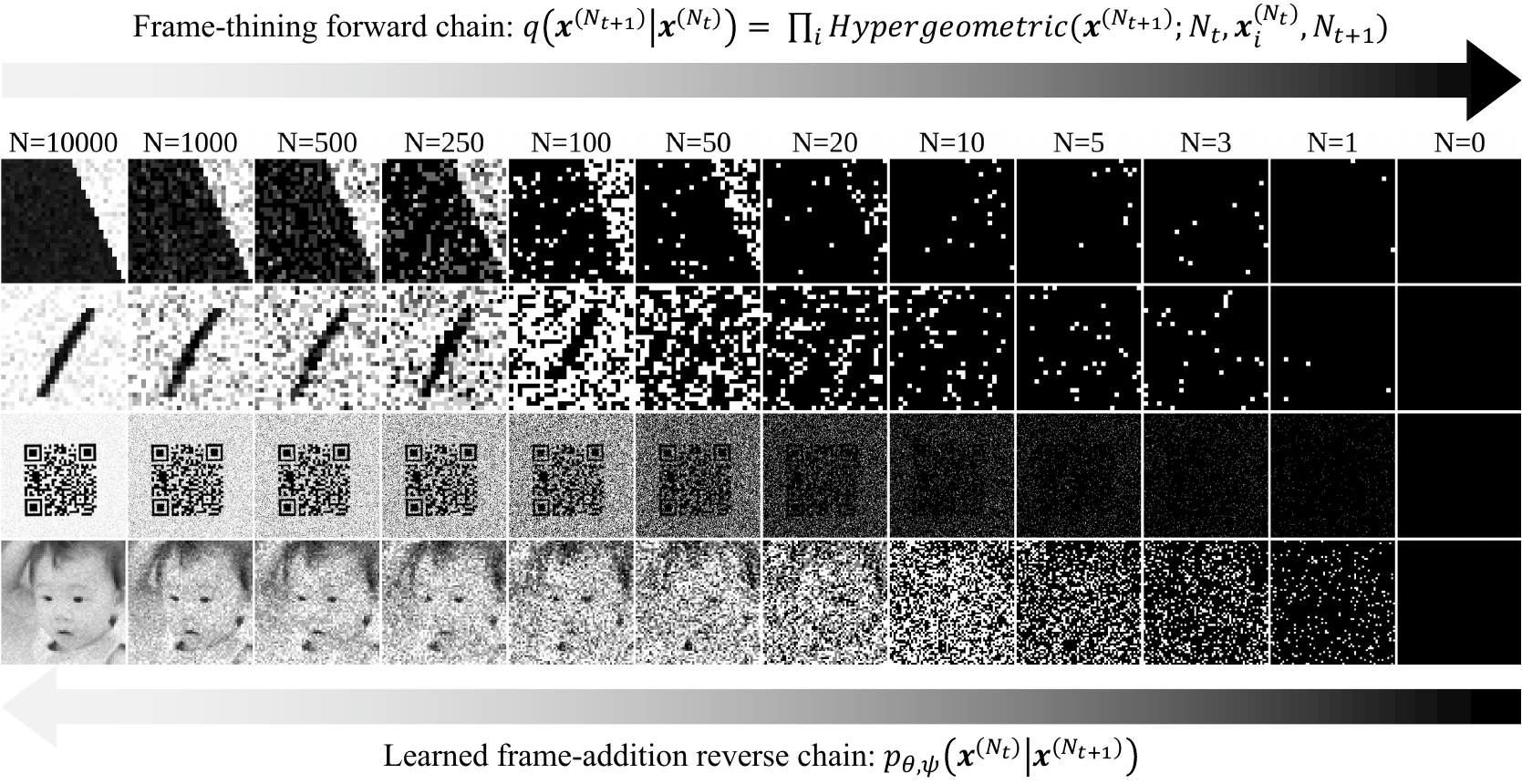}
    \caption{The overall process of conditional generator for SPAD image restoration. }
    \label{fig:conditional_gen_overview}
\end{figure*}

\sloppy
One can learn the conditional generator $q\!\left(\bX^{(N_1)} \mid \bX^{(N_k)} \right)$ in several ways.
Fixing $N_1$ to a  large value, we can train a conditional diffusion model \cite{dhariwal2021diffusion} from scratch or implement a posterior sampler (e.g., \cite{chung2023diffusion}) with a pre-trained diffusion model generate samples from estimated posterior distribution. 
However, since $N_k$ is expected to vary,  training individual generators for each $N_k$ would be computationally prohibitive.
Alternatively,  a single generator would need to span the full range from  unconditional generation ($N_k = 0$) to near-deterministic recovery ($N_k$ approaching $N_1$).

Figures~\ref{fig:mnist_random_spad_grid} and \ref{fig:conditional_gen_overview} suggests a more principled route based on the physics of image formation.
Consider the gamut of images that we observe as the photon budget changes.
At one end, we have photon-starved images and, at the limit, a blank image that has registered no photons at all; at the other end, we have well-exposed and noise-free images where every pixel has gathered enough photons so that its mean vastly outweights its standard deviation.
Traversing this photon ladder resembles the basic idea behind diffusion models: rather than map directly between two distinct distributions, one can learn to move along a sequence of simpler intermediate states \cite{sohl2015deep, ho2020denoising}. 
The difference is that our ladder is not constructed on an artificial Gaussian-noise schedule, but directly emerges from the  SPAD image formation process. 
Specifically, moving downward  the ladder corresponds to discarding binary frames---a frame-thinning operation that defines a \textit{Markov chain} and preserves the Binomial marginal distribution. 
This gives a natural reverse problem in the form of \textit{learning to add frames back.}
Our pairwise reverse model---which learns transitions between frames with consecutive number of measurements---learns this local reverse transition.
As we elaborate in Section \ref{sec:method}, our technique uses a small set of frame-count anchors, in the same  spirit as discrete diffusion models that learn reverse transitions over structured discrete states connected via a Markov chain \cite{austin2021structured}.

This conditional generator also provides an empirical measure of determinism/stochasticity, both in the context of image reconstruction and task-specific inference. 
Given a  measurement $\bX^{(N_k)} = \bx_k$, we can generate multiple samples from our learned approximation to $q\!\left(\bX^{(N_1)} \mid \bX^{(N_k)} =\bx_k \right)$.
The variation among these samples characterizes the breadth of the one-to-many map.
We further examine the stochasticity of these samples through the lens of specific downstream tasks
If the samples disagree in task-relevant ways,  the measurement remains stochastic at the task level. 
If the samples agree,  the measurement has become effectively deterministic. 
This viewpoint lets us study not just when an image can be made visually clean, but also how many binary SPAD frames are needed before downstream tasks become stable.
We study this transition across several applications: (1) regressing parameters of primitive edges from image patches, (2) an optical character recognition (OCR) proxy using the MNIST dataset,  (3) QR code reconstruction,  and (4)  generation of human faces with landmark and pose estimation. 
Across all applications, we vary the input frame count $N_k$ and characterize how many binary SPAD measurements are needed before the downstream decision and image reconstruction becomes stable.

%% file: method.tex
In this section, we describe the proposed method for learning the conditional 
generator, i.e., learn to sample from $q\!\left(\bX^{(N_1)} \mid \bX^{(N_k)} \right)$, where $N_1 \gg N_k$ results in $\bX^{(N_1)}$ being a photon-noise-free image and $\bX^{(N_k)}$ being a highly noisy image.
As mentioned before, our scope transcends the traditional denoising approach taken in prior work, in that, we tackle the one-to-many nature of the inverse problem.
%The goal is to learn the reverse conditional law from a darker measurement \(\bx^{(N_k)}\) to a brighter measurement \(\bx_t\), and eventually to the high-count target \(\bx^{(N_1)}\).

Throughout this section, \(\bX^{(N)}\) denotes a  SPAD image formed from the summation of $N$ binary frames and $\bx^{(N)}$ denotes a realization of this random variable. 
We will also adopt the convention commonly used in diffusion models to define a discrete time index $t$, where smaller values of $t$ correspond to ``cleaner''  images and larger values of $t$ correspond to  ``noisier'' images.
We realize this by linking each value of $t$ to a different frame count $N_t$, and hence $N_1$ takes on a very high value, and $N_t$ progressively decreases in value as $t$ increases,
i.e., $N_t > N_{t+1}.$
Finally, to simplify
 notation a bit,  we write
\[
\bX_t := \bX^{(N_t)}, \textrm{ and } \bx_t:=\bx^{(N_t)},
%\quad \bX_{t+1} := \bX^{(N_{t+1})},
\]
and denote the $i$-th pixel of each as $X_{t, i}$ and $x_{t, i}$, respectively.
%and similarly \(\bx_t:=\bx^{(N_t)}\), \(\bx_{t+1}:=\bx^{(N_{t+1})}\). 

\subsection{Markovian Frame-Thinning Forward Chain}

The SPAD measurement process naturally defines a ladder of images indexed by the number of accumulated binary frames.
Let us consider two consecutive images $\bx_t$ and $\bx_{t+1}$ on this ladder corresponding to frame counts $N_t$ and $N_{t+1}$, respectively.
Since \(N_t>N_{t+1}\), frame thinning corresponds to uniformly retaining \(N_{t+1}\) of the \(N_t\) underlying binary frames. The transition is conditioned only on the count image \(\bx_t\), not on the full binary-frame stack. At each pixel, the retained count follows a hypergeometric distribution:
\begin{equation}
\bX_{t+1,i}\mid \bX_{t,i} = \bx_{t,i}
\sim
\operatorname{Hypergeometric}(N_t,x_{t,i},N_{t+1}).
\end{equation}
This distribution describes the number of detections retained when sampling \(N_{t+1}\) binary outcomes without replacement from \(N_t\) outcomes containing \(x_{t,i}\) detections.
We assume that detections are independent across frames and pixels given \(\brho\). Applying this model across the entire image of \(D\) pixels, we get:
\begin{equation}
q(\bx_{t+1}\mid\bx_t)
=
\prod_{i=1}^D
\operatorname{Hypergeometric}(x_{t+1,i};N_t,x_{t,i},N_{t+1}).
\label{eq:forward-thinning-kernel}
\end{equation}

For completeness, we can also establish that the frame thinning process as described above and modeled by the hypergeometric distribution provides marginal distributions for $\bX_{t+1}$ that are Binomial at each pixel conditional on \(\brho\), as we would expect from a physical realization of the random variable based on Eq.~\eqref{eq:spad_pixel_measurement}.

\begin{theorem}[SPAD frame thinning is a Markov chain]
\label{thm:frame-thinning-markov}
Let
\[
N_1>N_2>\cdots>N_t>N_{t+1}>\cdots>N_L
\]
be a decreasing sequence of frame counts.
Successively applying the thinning transition
\[
q(\bx_{t+1}\mid\bx_t)
=
\prod_{i=1}^D
\operatorname{Hypergeometric}
\!\left(x_{t+1,i};N_t,x_{t,i},N_{t+1}\right)
\]
defines a Markov chain. Moreover, if
\[
\bX_t\mid\brho
\sim
\prod_i\operatorname{Binomial}(N_t,\rho_i),
\]
then the random variable after this thinning transition also follows a Binomial distribution of the form
\begin{equation}
\bX_{t+1}\mid\brho
\sim
\prod_i\operatorname{Binomial}(N_{t+1},\rho_i).
\label{eq:forward-thinning-binomial}
\end{equation}
\end{theorem}

The proof is given in the Appendix. The main point is simple: uniformly discarding frames preserves the physical Binomial marginal at the new frame count. Therefore, the resulting sequence of SPAD measurements at different \(N\) values is a sensor-defined Markov chain.

Importantly, the reverse of this Markov chain is also Markovian.
In short, our goal is to learn the reverse transition of the Markov chain to add frames back for conditional generation of the high-count image.

\subsection{Conditional Structure of the True Reverse Chain}
The reverse direction corresponds to the addition of missing binary frames. 
For \(N_t>N_{t+1}\), define
\begin{equation}
M_t := N_t-N_{t+1}.
%\qquad
%\Delta \bx_t := \bx_t-\bx_{t+1}.
\end{equation}
If the photon-detection probability  \(\brho\) were known, then the missing frames would be independent Bernoulli trials, so
\begin{equation}
\bX_t
\,=\,
\bX_{t+1} + \bY_t,
\ \ 
Y_{t,i} \mid \brho \sim \operatorname{Binomial}\!\big(M_t,\rho_i\big)
\text{ independently over $i$}.
\end{equation}
%Equivalently,
%\begin{equation}
%q(\bx_t\mid \bx_{t+1},\brho,N_t,N_{t+1})
%=
%\prod_{i=1}^D
%\operatorname{Binomial}(\Delta x_{t,i};M_t,\rho_i).
%\label{eq:true-pairwise-oracle}
%\end{equation}
In practice, \(\brho\) is unknown. Let
\[
\nu_{\mathrm{true}}(\brho\mid \bx_{t+1})
\]
denote the true posterior distribution of \(\brho\) after observing \(\bx_{t+1}\). Marginalizing \(\brho\) gives the exact reverse kernel:
\begin{equation}
\begin{aligned}
&q(\bx_t\mid \bx_{t+1}) \\ %; N_t,N_{t+1}) \\
& =
\int\!\!\left[
\prod_{i=1}^D
\operatorname{Binomial}(x_{t,i}-x_{t+1, i};M_t,\rho_i)
\right]\!
\nu_{\mathrm{true}}(\brho\mid \bx_{t+1})d\brho.
\end{aligned}
\label{eq:true-pairwise-mixture}
\end{equation}
%
%Equation~\eqref{eq:true-pairwise-mixture} is the key structure. 
%
Effectively, the true reverse kernel is the \textit{product of Binomial distributions}, integrated over the posterior of the hidden photon-probability image \(\brho\). 
Therefore, a learned reverse model should represent the uncertainty over \(\brho\), while keeping the final count-generation step Binomial.

\subsection{Latent Binomial Model}
\label{sec:latent-binomial-model}

Equation~\eqref{eq:true-pairwise-mixture} motivates a  reverse model with two parts: 
first, a latent variable that estimates $\nu_{\mathrm{true}}(\brho\mid \bX_{t+1})$, which  represents the remaining uncertainty on the photon-detection probability map $\brho$ conditioned on observing \(\bX_{t+1} = \bx_{t+1}\) and, second, a Binomial sampler that increases the frame count from $N_{t+1}$ to $N_t$.
%
%The latent model is not the only possible way to represent this mixture. For example, one could learn a conditional diffusion or flow model over the photon-probability field \(\brho\), which is just an exponential transformation of the image. 
%
Based on this, we train a reverse predictive kernel of the form:
\begin{equation}
p_{\theta,\psi}(\bX_t\mid \bX_{t+1} = \bx_{t+1}; N_t,N_{t+1}),
\end{equation}
%directly from SPAD count pairs.
where $\theta$ and $\psi$ denote trainable parameters of models that we describe below.

\subsubsection{The Reverse Kernel} 
The reverse kernel comprises two sampling steps built upon a prior encoder whose parameters are $\psi$ and a decoder with parameters $\theta$.
% For one reverse step, recall that
% \[
% M_t := N_t-N_{t+1},
% \qquad
% \Delta \bx_t := \bx_t-\bx_{t+1}.
% \]
The prior encoder takes as input $\bx_{t+1}$ along with the frame counts $N_{t+1}$ and $N_t$ to predict the parameters of the distribution of a latent variable $\bZ$.
Specifically, the encoder allows us to sample a latent vector $\bz \sim \bZ$ using a network with trainable parameters $\psi$:
\begin{equation}
\bz\sim p_\psi(\cdot \mid \bX_{t+1} = \bx_{t+1}; N_t,N_{t+1}).
\end{equation}
We will discuss the construction of $p_\psi$ below.
Given \(\bz\), the decoder predicts the photon-detection probability map as follows:
\begin{equation}
\bu = G_\theta(\bx_{t+1},N_t,N_{t+1},\bz),
\qquad
\brho_\theta=\sigma(\bu),
\end{equation}
where $\bu$ are the output logits of the network, and applying the sigmoid operator $\sigma$ gives the photon map $\brho_\theta$.
With the predicted photon-detection probability map \(\brho_\theta\), we reverse the SPAD frame thinning process by adding back \(M_t\) independent Bernoulli trials with probability \(\brho_\theta\):
\begin{equation}
    \bX_t
    \,=\,
    \bX_{t+1} + \widehat{\bY}_{t},
    \quad
    \widehat{Y}_{t,i} \mid \brho_\theta \sim \operatorname{Binomial}\!\big(M_t,\rho_{\theta,i}\big).
    \end{equation}
The conditional reverse kernel given \(\bz\) is
\begin{equation}
p_\theta(\bx_t\mid \bx_{t+1}, \bz)
=
\prod_{i=1}^D
\operatorname{Binomial}( x_{t,i}-x_{t+1, i};\ M_t,\rho_{\theta,i}).
\label{eq:learned-pairwise-kernel}
\end{equation}
Marginalizing the latent variable gives
\begin{equation}
\begin{aligned}
&p_{\theta,\psi}(\bx_t\mid \bx_{t+1}; N_t, N_{t+1}) =
\int
p_\theta(\bx_t \mid \bx_{t+1}, \bz)
\,p_\psi(\bz\mid \bx_{t+1})d\bz.
\end{aligned}
\label{eq:learned-pairwise-marginal}
\end{equation}
Thus the learned reverse kernel is also a mixture of product Binomial distributions. The important point is that the model is trained to approximate the reverse predictive distribution of \(\bX_t\), not necessarily to recover the full posterior distribution of \(\brho\) exactly.

% \paragraph{Implementation inputs.}
% In the implementation, the neural networks do not use raw frame counts directly. They use normalized images and frame-count features, for example
% \[
% r_{t+1}
% :=
% \frac{x_{t+1}}{\operatorname{safe}(N_{t+1})},
% \qquad
% \eta_{t+1}
% :=
% \log\!\left(\frac{\operatorname{safe}(N_{t+1})}{N_1}\right),
% \]
% \[
% \eta_t
% :=
% \log\!\left(\frac{N_t}{N_1}\right),
% \qquad
% \gamma_t
% :=
% \frac{N_t-N_{t+1}}{N_1}.
% \]
% We write \(G_\theta(x_{t+1},N_t,N_{t+1},z)\) for compactness.

%\subsubsection{Why We Train the Reverse Kernel Likelihood}
\subsubsection{Why not a direct approach?}
It is worth considering direct approaches to sample from the conditional distribution $q(\bX_{1} | \bX_{t})$. 
Variational autoencoders (VAEs) \citep{kingma2013auto} and generative adversarial networks (GANs) \cite{goodfellow2014generative} perform such sampling, but each suffers from  limitations in the form of blurred generations in the case of VAEs or stale ones in the case of GANs.

An alternative approach  would try to learn the posterior distribution
\[
\nu_{\mathrm{true}}(\brho\mid \bX_{t+1}; N_{t+1})
\]
and then sample the SPAD increment accordingly. This is a valid approach, but it requires the generator to model a broad family of clean probability-map distributions: for very small \(N_{t+1}\), the task is close to unconditional image generation, while for large \(N_{t+1}\), the task is close to deterministic reconstruction. 

Our model instead targets the local reverse kernel
\[
q(\bX_t\mid \bX_{t+1}; N_t,N_{t+1}).
\]
This is the distribution that we actually need for Markov-chain reversion. For a local step, the target is only the missing-frame increment, and not the full clean probability map. Therefore, the model can learn the aspects of the hidden photon-probability image that matter just for the next SPAD count transition.

The natural optimization objective  is
\begin{equation}
\max_{\theta,\psi}
\log p_{\theta,\psi}(\bx_t\mid \bx_{t+1}; N_t,N_{t+1}).
\label{eq:pairwise-likelihood-objective}
\end{equation}
However, the marginal likelihood in~Eq.\ \eqref{eq:learned-pairwise-marginal} contains an integral over \(\bz\). This integral is generally not available in closed form due to the complex form of the neural decoder. We therefore use a variational training objective.

\subsubsection{Variational Training with a Posterior Encoder}
We follow the standard VAE \cite{kingma2013auto} and conditional VAE \cite{sohn2015learning} approach: a prior network is used at sampling time, and a posterior encoder is used only during training to provide useful latent samples. A schematic of the VAE architecture is shown in Figure \ref{fig:latent_binomial_model}.

We parameterize the prior as a Gaussian with diagonal covariance matrix,
\begin{equation}
\begin{aligned}
(\bmu^p,\bell^p)
&=
E^{\mathrm{prior}}_\psi(\bx_{t+1},N_t,N_{t+1}),\\
p_\psi(\bz \mid \bx_{t+1},N_t,N_{t+1})
&=
\mathcal N\!\left(\bz ; \bmu^p,\operatorname{diag}(e^{\bell^p})\right).
\end{aligned}
\label{eq:pairwise-prior}
\end{equation}
During training, the posterior encoder also sees the target through the increment image:
\begin{equation}
\bd_t := \frac{\bx_t-\bx_{t+1}}{N_t-N_{t+1}}.
\label{eq:pairwise-increment-feature}
\end{equation}
It predicts
\begin{equation}
\begin{aligned}
(\bmu^q,\bell^q)
&=
E^{\mathrm{post}}_\phi(\bx_{t+1},\bd_t,N_t,N_{t+1}),\\
q_\phi(\bz\mid \bx_t,\bx_{t+1},N_t,N_{t+1})
&=
\mathcal N\!\left(\bz;\bmu^q,\operatorname{diag}(e^{\bell^q})\right).
\end{aligned}
\label{eq:pairwise-posterior}
\end{equation}

\begin{figure}[t]
    \centering
    \includegraphics[width=1\columnwidth]{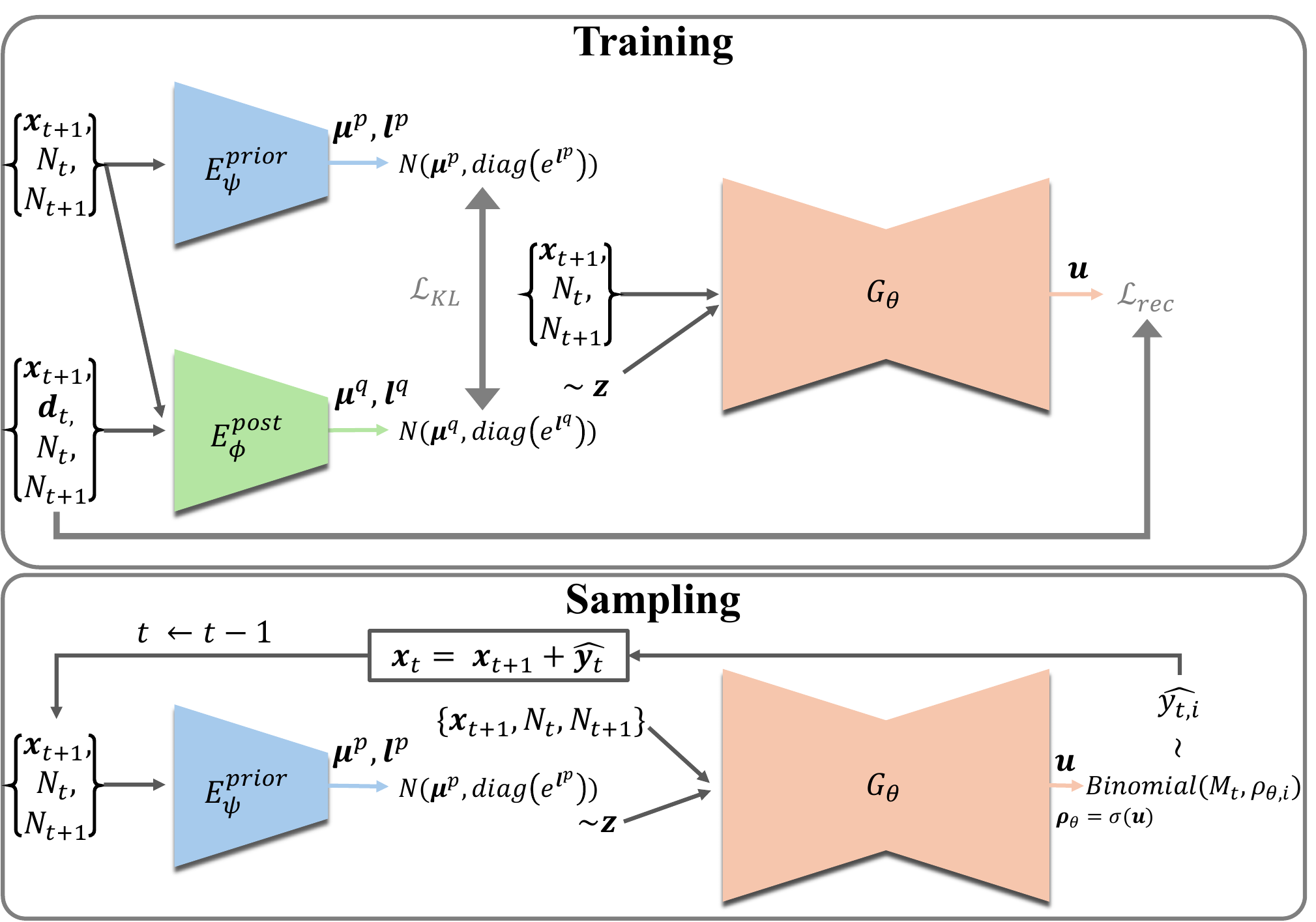}
    \caption{Overview of the proposed latent binomial model for a single step. The prior encoder takes in the source SPAD image \(\bx_{t+1}\) and the frame counts \(N_t\) and \(N_{t+1}\) to predict the parameters of the Gaussian prior distribution of the latent variable \(\bz\). The posterior encoder takes in the source SPAD image \(\bx_{t+1}\), the increment image \(\bd_t\), and the frame counts \(N_t\) and \(N_{t+1}\) to predict the parameters of the Gaussian posterior distribution. The decoder with a U-Net architecture takes in \(\bx_{t+1}\), the frame counts \(N_t\) and \(N_{t+1}\), and \(\bz\) to predict the photon-detection probability map \(\brho_\theta\). During sampling, \(\bz\) is drawn from the prior, \(\widehat{y}_{t,i}\sim\operatorname{Binomial}(M_t,\rho_{\theta,i})\), and \(\bx_t=\bx_{t+1}+\widehat{\bm{y}}_t\). Here \(\bz\) denotes the set of latent groups described in Section~\ref{sec:latent-binomial-model}.}
    \label{fig:latent_binomial_model}
\end{figure}

The pairwise evidence lower bound is
\begin{equation}
\log p_{\theta,\psi}(\bx_t\mid \bx_{t+1},N_t,N_{t+1})
\ge
\mathcal L_{\mathrm{pair}}(\bx_t,\bx_{t+1}),
\label{eq:pairwise-elbo-bound-method}
\end{equation}
where
\begin{equation}
\begin{aligned}
& \mathcal L_{\mathrm{pair}}
={}
\mathbb E_{q_\phi}
\left[
\log p_\theta(\bx_t\mid \bx_{t+1},N_t,N_{t+1},\bz)
\right]
\\
&
\quad
-\operatorname{KL}
\left(
q_\phi(\bz\mid \bx_t,\bx_{t+1},N_t,N_{t+1})
\,\|\, 
p_\psi(\bz\mid \bx_{t+1},N_t,N_{t+1})
\right).
\end{aligned}
\label{eq:pairwise-elbo-method}
\end{equation}
The posterior encoder helps training because it can use \(\bx_t\) to infer latent variables that explain the observed target. The KL term keeps this training-time posterior close to the prior that will be used at test time.

\paragraph{Hierarchical latent implementation.}
For clarity, the equations use one symbol \(\bz\), but in practice we implement a hierarchical latent structure following the general direction of hierarchical VAEs \cite{sonderby2016ladder, vahdat2020nvae, child2020very}. The \(\bz\) symbol now denotes a collection of latent groups, such as
\[
\bz=(\bz_g,\bz_{s_1},\bz_{s_2}),
\]
where \(\bz_g\) is a global latent and \(\bz_{s_1},\bz_{s_2}\) are spatial latent maps at decoder resolutions \(s_1\) and \(s_2\). This lets the model represent uncertainty at multiple spatial scales. The prior and posterior encoders predict the parameters of \(\bz_g\) only. The prior and posterior parameters of each spatial latent are predicted inside the decoder from its features at that resolution, and the posterior also uses the posterior-encoder features at the same resolution. Hence, each spatial latent is conditioned on \(\bz_g\) and on the spatial latents at coarser resolutions.
% : the global latent captures coarse scene-level ambiguity, while spatial latents capture local structure. 
The ELBO has the same form after reading \(\bz\) as this compound latent variable. The KL term is computed as the sum of groupwise Gaussian KL terms, each using the same diagonal-Gaussian formula. 
More technical details are provided in Appendix.

\subsubsection{Practical Pairwise Loss}

Using the Binomial logit form and dropping the combinatorial term that does not depend on the model parameters, the reconstruction loss is simplified as
\begin{equation}
\mathcal J_{\mathrm{rec}}(\bx_t,\bx_{t+1})
=
\mathbb E_{\beps}
\left[
\sum_{i=1}^D
\left(
M_t\operatorname{softplus}(u_i)
-
\Delta x_{t,i}u_i
\right)
\right],
\label{eq:pairwise-rec-loss-method}
\end{equation}
where \(\Delta \bx_t := \bx_t-\bx_{t+1}\) and
\[
\beps\sim\mathcal N(\bm{0},\bm{I}),
\qquad
\bz=\bmu^q+\exp\!\left(\tfrac12\bell^q\right)\odot\beps .
\]
The factor \(1/2\) appears because \(\bell^q\) is the log-variance, so
\[
\exp\!\left(\tfrac12\bell^q\right)
\]
is the standard deviation. In practice, this expectation is estimated with one posterior latent sample per training example.

The Gaussian KL term has a closed form. For a single latent group,
\begin{equation}
\mathcal J_{\mathrm{KL}}
=
\frac12
\sum_r
\left[
\ell_r^p-\ell_r^q-1
+
e^{\ell_r^q-\ell_r^p}
+
(\mu_r^q-\mu_r^p)^2e^{-\ell_r^p}
\right].
\label{eq:pairwise-gaussian-kl}
\end{equation}
% For hierarchical latents, this expression is summed over all coordinates of each latent group and then over groups. 
The practical negative ELBO is
\begin{equation}
\mathcal J_{\mathrm{pair}}
=
\mathcal J_{\mathrm{rec}}
+
\mathcal J_{\mathrm{KL}}.
\label{eq:pairwise-practical-loss}
\end{equation}
For hierarchical latents, the KL term is decomposed into groups,
\[
\mathcal J_{\mathrm{KL}}
=
\sum_{g\in\mathcal G}\mathcal J_{\mathrm{KL},g},
\]
where \(\mathcal G\) indexes the global and spatial latent groups. In practice, we use groupwise KL weights,
\begin{equation}
\mathcal J_{\mathrm{pair}}
=
\mathcal J_{\mathrm{rec}}
+
\sum_{g\in\mathcal G}
\beta_g(s)\,
\mathcal J_{\mathrm{KL},g},
\label{eq:pairwise-beta-kl-loss}
\end{equation}
where \(s\) is the optimizer step. The weights \(\beta_g(s)\) are increased during training to avoid suppressing the latent variables too early, a common issue in VAE training \cite{sonderby2016train}. Before the sum over \(r\) in~\eqref{eq:pairwise-gaussian-kl}, each coordinate term is floored at a constant \(\lambda\) (free bits). 
% This is similar in spirit to KL warm-up and KL annealing methods used for training deep or conditional VAEs.

When the training set contains transitions with very different frame gaps \(M_t\), we also use a gap-normalized reconstruction term:
\begin{equation}
\widetilde{\mathcal J}_{\mathrm{pair}}
=
\frac{1}{M_t}\mathcal J_{\mathrm{rec}}
+
\sum_{g\in\mathcal G}
\beta_g(s)\,
\mathcal J_{\mathrm{KL},g}.
\label{eq:gap-normalized-method}
\end{equation}
This normalization only reweights the reconstruction term across transitions of different sizes. It does not change the generative model in~\eqref{eq:learned-pairwise-kernel}. Its role is to prevent large frame-gap transitions from dominating the minibatch loss.
More details and proofs are provided in the Appendix.

\subsection{Anchor-point Markov Reverse Chain for Arbitrary Number of Measurements}
\label{sec:anchor-point-reverse-chain}

To reverse the Markovian forward chain, the ideal approach would learn every adjacent reverse kernel (i.e., single binary frame addition) in the ladder, i.e., have $N_{t} = N_{t+1} + 1$.
In practice, using many small steps would dramatically slow the generation process, as we 
expect to accumulate thousands if not tens of thousands of SPAD frames before the effect of photon noise is minuscule. 
%especially an issue for our purpose of measuring the transition of measurement determinism.
We therefore use a short Markov reverse chain defined by a small set of anchor frame counts.
Let
\[
\mathcal{A} =\{A_1=N_1, A_2,\dots, A_L\},\]
be a list of monotonically decreasing integers that define the values that $N_t$ take.
For any integer $N_k \in [0, N_1)$, we can define the next anchor
\begin{equation}
T(N_k)=\min\{A \in\mathcal A \textrm{ s.t. } A > N_k\}.
\label{eq:generic-next-anchor-method}
\end{equation}
Hence, if we start with a  SPAD image measured with $N_k$ binary frames, then we set a reverse chain with frame counts
\[ N_k \rightarrow T(N_k) \rightarrow T(T(N_k)) \rightarrow \cdots \rightarrow  N_1.\]
This allows us to progressively sample from the pairwise reverse kernel to transform low frame count SPAD measurements into high frame count ones.
%Starting from a measured SPAD image $\bX_{t+1} = \bx_{t+1}$ for a frame count $N_{t+1} \in \mathcal{N}$,
% Let
% \[
% \mathcal A=\{A_1,A_2,\dots,A_L\},
% \qquad
% A_1=N_1,
% \]
% be a monotically decreasing list of  integers, that define the values that $N_t$ take.
% %
% Given a current frame count \(n > N_1\), define the next anchor by
% \begin{equation}
% T(n)=\max\{A_r\in\mathcal A:A_r < n\}.
% \label{eq:generic-next-anchor-method}
% \end{equation}
% Starting from a measured SPAD image
% \[
% (\bx_L,N_L = A_L)
% \]
% we define the reverse frame-count path by
% \[
% N_{t} = T(N_{t+1}).
% \]
At each step, we sample from the shared pairwise reverse kernel:
\begin{equation}
\bx_{t}
\sim
p_{\theta,\psi}
\left(
\cdot
\mid
\bx_{t+1};  N_{t}, N_{t+1}
\right).
\label{eq:anchor-sampling-step}
\end{equation}
Repeating the process, we trace out SPAD measurements with frame count in $\mathcal{A}$ until we reach $N_1$, and the final state is the generated high-count SPAD image.

This anchor chain is a compressed version of the full Markov reverse process. The anchors are chosen to balance two goals. Smaller steps make each local reverse problem closer to the physical frame-addition transition, while fewer steps keep sampling efficient and reduce the accumulation of local errors. The same pairwise model is trained and shared across all anchor transitions, so the method does not require a separate network for each edge of the chain.

%\subsection{Algorithm}

We summarize the algorithm for training the latent Binomial model in Algorithm~\ref{alg:pairwise-training-method}, and for sampling in Algorithm~\ref{alg:pairwise-sampling-method}.
%In this section, we provide the detailed algorithm for training the latent Binomial model with the pairwise loss in Equation~\eqref{eq:gap-normalized-method}, and sampling by anchor-point Markov reversion described the previous section.

\begin{algorithm}[t]
\caption{Training the shared pairwise SPAD reverse model}
\label{alg:pairwise-training-method}
\begin{algorithmic}[1]
\Require Clean-image dataset; target frame count \(N_1\); anchor frame counts $\mathcal{A}$
\Repeat
    \State Sample a clean image \(\s_0\)
    \State Compute \(\brho=1-\exp(-\tau \s_0)\)
    \State Sample $N_k$ from $[0, \ldots, N_1)$
    \State Set source frame count: \(N_{t+1} \leftarrow N_k\) 
    \State Set target frame count: $N_t \leftarrow T(N_k)$ as defined in Eq.\ \eqref{eq:generic-next-anchor-method}
    \State Sample
    \begin{align*}
    \bx_{t} &\sim \prod_i\operatorname{Binomial}(N_{t},\rho_i), \\
    % \Delta \bx_t &\sim \prod_i\operatorname{Binomial}(N_t-N_{t+1},\rho_i)
    \bx_{t+1} &\sim \prod_i\operatorname{Hypergeometric}(N_t,x_{t,i},N_{t+1})
    \end{align*}
    % \State Set \(\bx_t=\bx_{t+1}+\Delta \bx_t\)
    % \State Form \(r_{t+1},\bd_t,\eta_{t+1},\eta_t,\gamma_t\)
    \State Predict prior and posterior Gaussian parameters using Eq.\ \eqref{eq:pairwise-prior} and \eqref{eq:pairwise-posterior}.
    \State Sample \(\bz\sim q_\phi(\cdot\mid \bx_t,\bx_{t+1},N_t,N_{t+1})\)
    \State Decode logits \(\bu\) and compute the pairwise loss in Eq.\ \eqref{eq:gap-normalized-method}
    \State Update \((\theta,\psi,\phi)\)
\Until{convergence}
\end{algorithmic}
\end{algorithm}

\begin{algorithm}[t]
\caption{Sampling by anchor-point Markov reversion}
\label{alg:pairwise-sampling-method}
\begin{algorithmic}[1]
\Require Input SPAD image \(\bx_{inp} \) corresponding to input frame count \(N_k\), Anchor frame counts $\mathcal{A}$
\State Determine $t$ such that the $t$-th element of $\mathcal{A}$ is $T(N_k)$
\State $\bx_{t+1} \leftarrow \bx_{inp}$
\State $N_{t+1} \leftarrow N_k$
%\State Set \(\bx^{(0)}\leftarrow \bx^{(N_k)}\), \(N^{(0)}\leftarrow N_k\), \(\ell\leftarrow0\)
\While{$t > 0$}
%\While{\(N^{(\ell)}<N_1\)}
%    \State Set \(N^{(\ell+1)}\leftarrow T(N^{(\ell)})\)
    \State $N_t = T(N_{t+1})$
    \State  \(M_t \leftarrow N_t - N_{t+1} \)
    \State Sample latent variable
    \[
    \bz
    \sim
    p_\psi
    \left(
    \cdot
    \mid
    \bx_{t+1},N_{t}, N_{t+1}
    \right)
    \]
    \State Decode logits and probability map
    \[
    \bu
    \leftarrow
    G_\theta\!\left(\bx_{t+1},N_t, N_{t+1},\bz\right),
    \qquad
    \brho\leftarrow \sigma\!\left(\bu\right)
    \]
    \State Sample Binomial count increment independently for each pixel \(i\):
    \[
    \widehat{y}_i
    \sim
    \operatorname{Binomial}\!\big(M_t,\rho_i\big)
    \]
    \State  \(\bm{x}_t \leftarrow \bx_{t+1}+\widehat{\bm{y}}\)
%    \State Set \(\ell\leftarrow \ell+1\)
    \State  $t \leftarrow t - 1$

\EndWhile
\State Return $\bx_1$
\end{algorithmic}
\end{algorithm}

%% file: cases.tex
\section{Case Study: Edges}
\label{sec:edges}
Edges are primitive structures in natural images and provide a controlled setting for studying how geometric information becomes determined as the photon budget increases. In this case study, we generate synthetic straight-edge images, train the proposed conditional SPAD generator, and evaluate whether generated high-count samples agree on task-relevant edge attributes: orientation and location.

\subsection{Experimental settings}

\paragraph{Synthetic edge dataset.}
We construct a synthetic dataset of two-tone straight-edge patterns on \(28\times28\) grayscale images. Each image is defined on normalized coordinates \((x,y)\in[-1,1]^2\). For each sample, we draw an edge orientation \(\theta\sim\mathcal{U}(0,\pi)\) and a point \(\mathbf{c}\sim\mathcal{U}([-1,1]^2)\). The edge normal is
\(\mathbf{n}=(-\sin\theta,\cos\theta)\), giving the separating line
\(\mathbf{n}^{\top}\mathbf{x}=b\), where \(b=\mathbf{n}^{\top}\mathbf{c}\).
Pixels on a randomly chosen side of the line are uniformly assigned to intensities of \(1.0\). The other side of the line are assigned to \(0.1\) for a realistic setting of getting photons from the dark side as well. The ground-truth attributes are the orientation \((\cos\theta,\sin\theta)\) and the line bias \(b\).

\paragraph{SPAD simulation.}
As described in Section~\ref{sec:overview-one-to-many}, each pixel count is sampled from the Binomial SPAD model in Eq.~\eqref{eq:spad_pixel_measurement}, with photon-detection probability \(\rho_i\) defined in Eq.~\eqref{eq:rho-def}. We set the light level to \(\tau=0.01\) and use \(N_1=10{,}000\) frames as the high-count target, to ensure few measurements don't contain too much information about the scene. During training, conditioning measurements are sampled with \(N_k\in[0,500]\). The anchor counts we use are 
\[
\mathcal A = \{10000, 500, 20\}.
\]
%before the final target \(N_1\). 
% At test time, we evaluate input measurements with
% \[
% N_k \in \{0,1,3,5,10,20,50,100,250\}.
% \]

\subsection{Results}
\paragraph{Qualitative generation.}
For each frame count \(N_k\), we fix one conditioning measurement \(\bx_k\) and generate 200 samples from the conditional generator. Figure~\ref{fig:edge_gen_sample_grid} shows the conditioning input, representative generated samples, and the sample mean and standard deviation maps.

When \(N_k=0\), the model receives no scene information and behaves as an unconditional generator, producing edges with different orientations and locations. As \(N_k\) increases, the generated samples become increasingly consistent. In the shown examples, even a few frames strongly constrain the edge geometry. The standard deviation maps further show that uncertainty becomes concentrated near the edge boundary, indicating that the remaining ambiguity is mainly in precise edge location rather than global structure.

\begin{figure}[t]
    \centering
    \includegraphics[width=1\columnwidth]{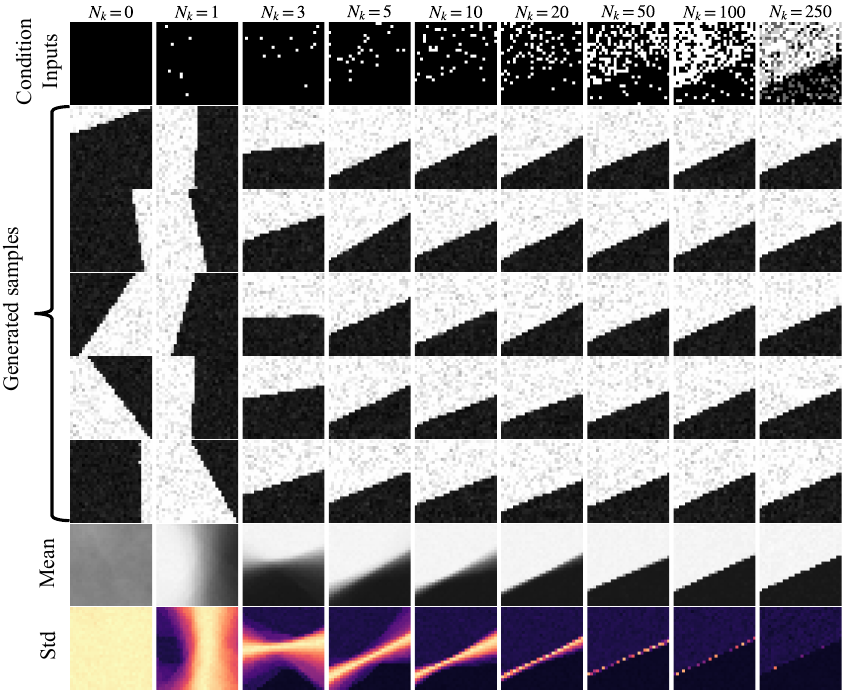}
    \caption{Conditional generation results for synthetic edge images. Each column fixes the same conditioning measurement \(\bx_k\) at a given frame count \(N_k\), followed by several high-count samples \(\bx_1\) generated conditioned on \(\bx_k\), their sample mean, and their sample normalized standard deviation.}
    \label{fig:edge_gen_sample_grid}
\end{figure}

\paragraph{Downstream task: edge prediction.}
To quantify task-level determinism, we train a CNN predictor on high-count target SPAD images \(\bx_1\) with \(N_1=10{,}000\) frames. The predictor estimates the edge orientation and line bias. We then generate a test set of 1,000 random edge images. For each image and each \(N_k\), we sample one conditioning measurement \(\bx_k\), generate high-count samples, and evaluate the CNN predictor on the generated outputs.

Figure~\ref{fig:edge_regression_resutls_vs_Nk} reports the angle error and bias error as functions of \(N_k\). The red dotted line shows the predictor performance on clean high-count SPAD images and serves as the performance ceiling. When a metric approaches this ceiling, the generated samples are effectively deterministic for that downstream attribute.
The results show that different edge attributes become deterministic at different photon budgets. Bias estimation stabilizes early, around \(N_k=20\), while accurate orientation estimation require roughly \(N_k=100\) frames. 
% Thus, deterministic inference emerges at different rates across edge parameters: the line offset becomes stable before the orientation is fully resolved.

\paragraph{Results on real data.}
To validate the proposed method on real measurements, we printed a straight-edge target and captured it with a Pi Imaging \(512\times512\) SPAD camera. The captured SPAD measurements are downsampled to \(28\times28\) pixels to match the resolution used by the conditional generator.

Real SPAD sensors contain hot pixels, i.e., pixels with elevated dark-count rates that can produce detections even when no light is incident on the sensor. To suppress these artifacts, we first captured 100{,}000 dark frames with the sensor optically blocked and used them to estimate a bad-pixel mask. Pixels marked by this mask are excluded during the downsizing process.

Figure~\ref{fig:real_edge_results} shows the result on the real captured edge. The generated samples exhibit quality comparable to the simulated edge results in Figure~\ref{fig:edge_gen_sample_grid}, indicating that the proposed model can generalize to real SPAD measurements after simple hot-pixel masking and downsizing.

\begin{figure}[t]
    \centering
    \includegraphics[width=1\columnwidth]{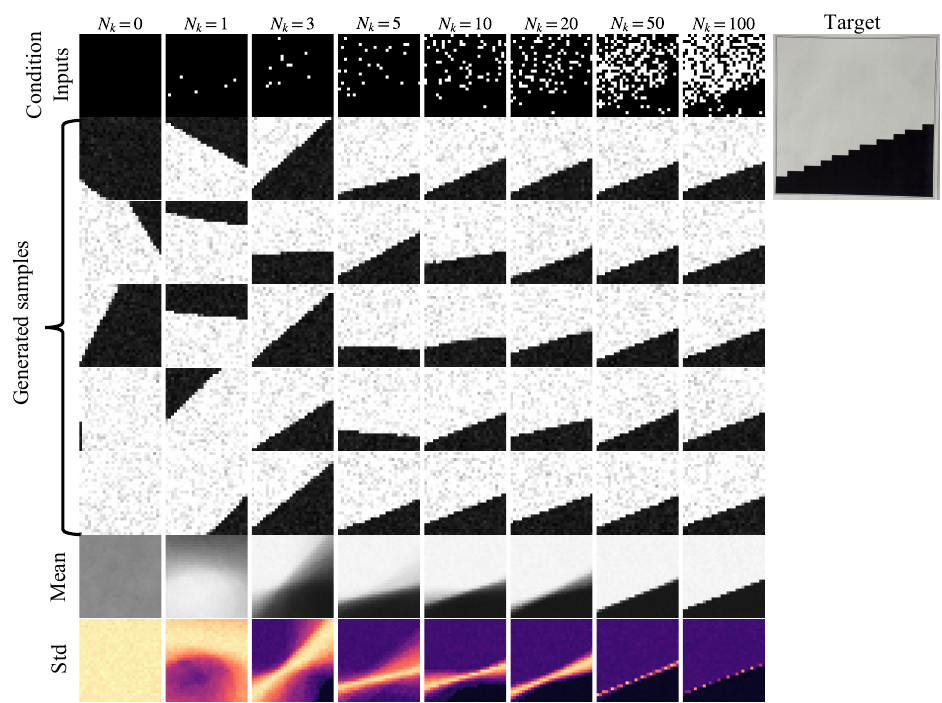}
    \caption{Conditional generation results on a real SPAD capture of a printed edge target. For each column, the figure shows the generated samples, their sample mean, and their sample standard deviation. At the rightmost corner, we show the photo of the printed target.}
    \label{fig:real_edge_results}
\end{figure}

\begin{figure}[t]
    \centering
    \includegraphics[width=1\columnwidth]{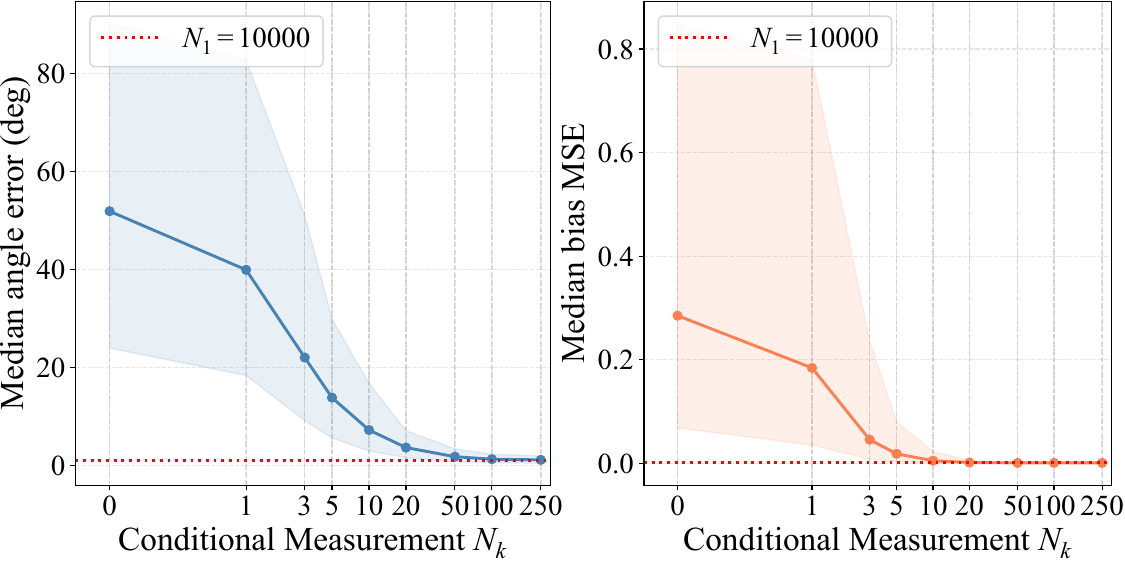}
    \caption{Downstream edge-prediction results for generated samples as a function of the number of conditioning SPAD frames \(N_k\). The red dotted line denotes performance on high-count target SPAD images with \(N_1=10{,}000\) frames.}
    \label{fig:edge_regression_resutls_vs_Nk}
\end{figure}

%===========================================
% \begin{figure*}[t]
%     \centering
%     \includegraphics[width=0.9\textwidth]{figs/MNIST_gen_sample_grid.pdf}
%     \caption{Conditional generation results for MNIST. For each fixed conditioning measurement \(\bx_k\), the figure shows generated high-count samples, their sample mean, and their sample standard deviation.}
%     \label{fig:mnist_gen_sample_grid}
% \end{figure*}

\section{Case Study: Optical Character Recognition}
\label{sec:mnist}
% Optical Character Recognition (OCR) has been a long-standing problem in computer vision. In this case study, we built a synthetic digit dataset based on the MNIST dataset~\cite{lecun2002gradient} and trained the proposed conditional SPAD generator to generate digits with given SPAD measurements.
% We then evaluate the classification accuracy of the generated samples by the number of binary SPAD measurements to identify the number of frames needed for the generated samples to become deterministic for the downstream task.

We next study optical character recognition (OCR), where the downstream task depends primarily on recovering digit identity rather than exact pixel values. Using MNIST~\cite{lecun2002gradient}, we construct a synthetic low-light digit dataset, train the proposed conditional SPAD generator, and measure how classification accuracy changes with the number of conditioning SPAD frames. This experiment lets us identify the photon budget at which digit identity becomes stable under conditional generation.

\subsection{Experimental settings}
\paragraph{Synthetic digit dataset.}
We construct the digit dataset by inverting MNIST images so that dark digits appear on a bright background, mimicking printed digits on paper. The images are then rescaled so that the maximum intensity is \(1.0\) and the minimum intensity is \(0.05\). This nonzero floor allows the dark digit regions to still produce photon detections, making the simulation more realistic than using completely black strokes. All images have resolution \(28\times28\) pixels.

\paragraph{SPAD simulation.}
We simulate SPAD measurements using the Binomial model in Eq.~\ref{eq:spad_pixel_measurement}. The light level is set to \(\tau=0.01\), and the high-count target is \(N_1=10{,}000\) frames. During training, conditioning measurements are sampled with \(N_k\in[0,500]\). We use the intermediate anchor counts \(
\mathcal A = \{N_1, 500, 100, 20, 10\}.
\)

\subsection{Results}
\paragraph{Qualitative generation.}
Figure~\ref{fig:mnist_gen_sample_grid} shows conditional generations from fixed SPAD measurements. For each conditioning input \(\bx_k\), we draw 200 samples and visualize representative outputs together with their sample mean and standard deviation.

At very small \(N_k\), the same conditioning measurement can produce samples from different digit classes, reflecting strong uncertainty in digit identity. As \(N_k\) increases, the samples concentrate on the target digit class. In the examples shown, most samples already agree with the target identity by \(N_k=20\). The standard deviation maps show that residual variation is concentrated around digit strokes and boundaries, suggesting that class-level ambiguity resolves before fine stroke location and shape become fully fixed.

\begin{figure}[t]
    \centering
    \includegraphics[width=1\columnwidth]{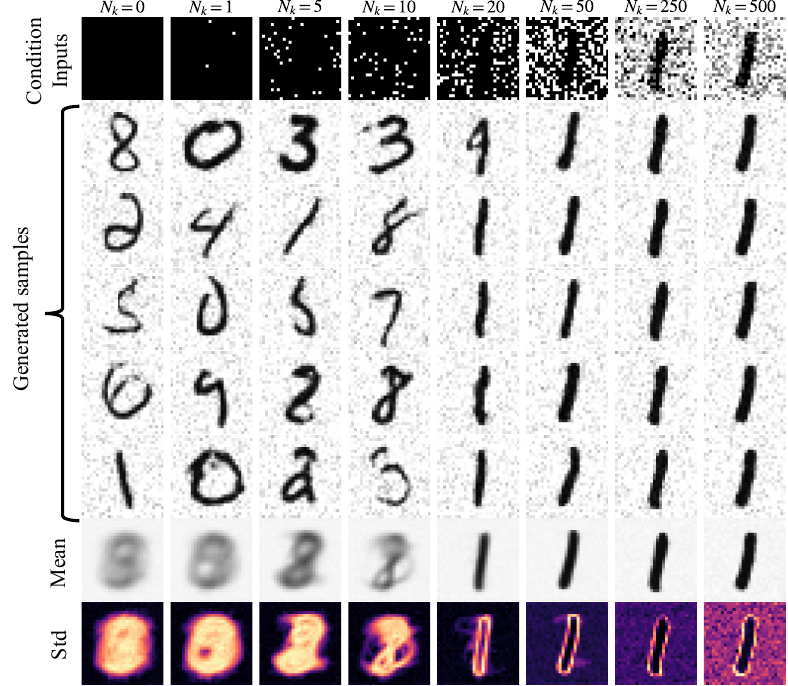}
    \caption{Conditional generation results for MNIST. For each fixed conditioning measurement \(\bx_k\), the figure shows generated high-count samples, their sample mean, and their sample standard deviation.}
    \label{fig:mnist_gen_sample_grid}
\end{figure}

\paragraph{Downstream task: digit classification.}
We quantify OCR-level determinism using a CNN classifier trained on high-count target SPAD images \(\bx_1\) with \(N_1=10{,}000\) frames. We form a test set of 1,000 digit images, uniformly distributed across the ten digit classes. For each image and each \(N_k\), we sample one conditioning measurement \(\bx_k\), generate high-count samples, and classify the generated outputs.

Figure~\ref{fig:avg_MNIST_classification_accuracy_vs_Nk} reports the average classification accuracy as a function of \(N_k\). Accuracy increases steadily with the conditioning frame count and exceeds \(99\%\) by \(N_k=250\).
% Figure~\ref{fig:cm_cls_gen_mnist} shows the corresponding confusion matrices. As early as \(N_k=10\), the diagonal entries become visibly dominant, indicating that the generated samples increasingly preserve the correct digit identity.

\paragraph{Results on real data.}
We also evaluate the MNIST model on real SPAD measurements. We printed an \(18\times18\) grid of digit images and captured it with the SPAD camera. As in the edge experiment, hot pixels are identified from a dark calibration capture and excluded during preprocessing.

To align the printed grid with the model input, we manually clicked the four corners of the digit grid in a high-count reference image and estimated a homography. Using this homography, each captured SPAD count image is rectified onto a \(504\times504\) grid, corresponding to \(18\times18\) digit cells with \(28\times28\) pixels per cell. The rectified image is then partitioned into \(28\times28\) digit patches, which are used as conditioning inputs to the conditional generator.

Figure~\ref{fig:ocr_real_results} shows restoration results for different numbers of SPAD measurements. As the number of measurements increases, more generated digits agree with the printed target grid, demonstrating the same stochastic-to-deterministic transition observed in simulation.

\begin{figure}[t]
    \centering
    \includegraphics[width=0.9\columnwidth]{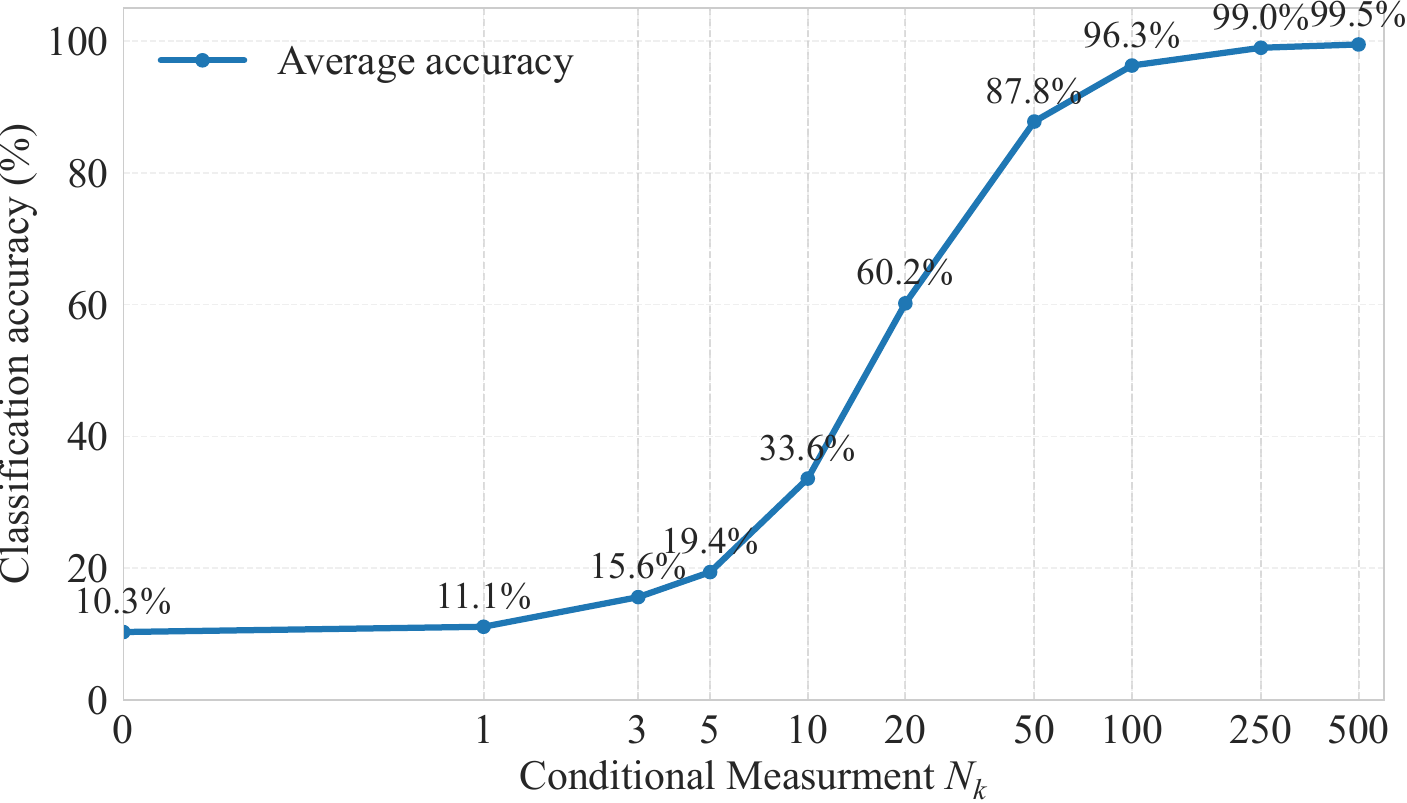}
    \caption{Classification accuracy of generated MNIST samples as a function of the number of conditioning SPAD frames \(N_k\).}
    \label{fig:avg_MNIST_classification_accuracy_vs_Nk}
\end{figure}

% \begin{figure}[t]
%     \centering
%     \includegraphics[width=1\columnwidth]{figs/cm_cls_gen_mnist.pdf}
%     \caption{Confusion matrices for generated MNIST samples across different conditioning frame counts \(N_k\).}
%     \label{fig:cm_cls_gen_mnist}
% \end{figure}

\begin{figure*}[t]
    \centering
    \includegraphics[width=1\textwidth]{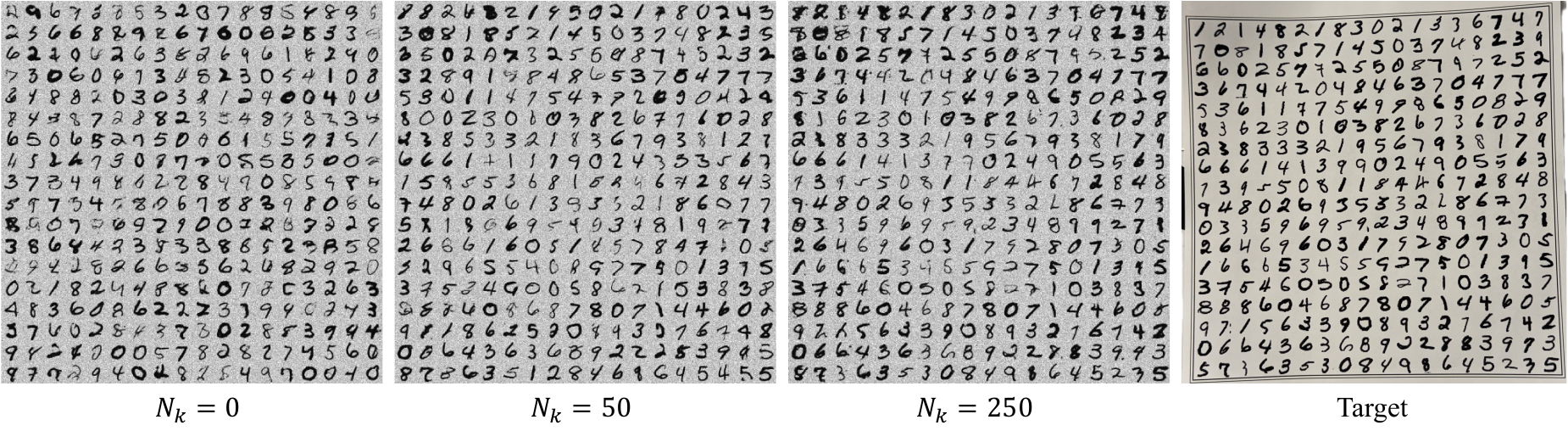}
    \caption{Conditional generation results with real captures in OCR task. We exhibit the restored grid of digits with selected number of measurements.}
    \label{fig:ocr_real_results}
\end{figure*}

%===========================================
\section{Case Study: QR Code Readers}
\label{sec:QR}
QR codes are widely used to encode information in a compact and machine-readable form. 
Decoding QR codes under extreme low-light conditions is therefore important for applications such as real-time warehouse sorting and traceability~\cite{chen2019adaptive}, high-throughput conveyor-belt product tracking~\cite{qian2017optimization} where short exposure times are required to avoid motion blur.
In this case study, we use a public QR-code dataset, simulate SPAD measurements, train the proposed conditional SPAD restoration model, and evaluate decoding accuracy as a function of the number of binary SPAD frames. We also compare single-step and multistep reverse generation to evaluate the benefit of intermediate anchor frames.

\subsection{Experimental settings}
\paragraph{Figshare QR code dataset.}
We use a subset of the Figshare QR Code Dataset V2~\cite{sancar2025}, which contains computer-generated QR codes with distortions designed to mimic scanner-induced effects. Each image is a single-channel grayscale QR code of size \(224\times224\). We use 12,799 images for training and 3,200 images for validation/testing. The encoded payload lengths are evenly distributed across 12, 24, 36, and 48 characters. The payload strings are alphanumeric strings resembling hexadecimal sequences, where each character corresponds to four bits if interpreted as a hexadecimal symbol.

\paragraph{SPAD simulation.}
As in the previous case studies, each pixel count is sampled from the Binomial SPAD model. We set the light level to \(\tau=0.01\) and use \(N_1=10{,}000\) frames as the high-count target. During training, conditioning measurements are sampled with \(N_k\in[0,500]\). The multistep model uses intermediate anchor counts \(\mathcal A = \{N_1, 100, 20, 10 \}\).
% We evaluate the generated samples by the number of binary SPAD measurements with \(N_k \in \{0, 1, 3, 5, 7, 10, 15, 20, 35, 50, 100, 250, 500\}\).

\subsection{Results}

\paragraph{Qualitative stochasticity of multistep generation.}
We next visualize the conditional samples produced by the multistep model. Figure~\ref{fig:qr_gen_sample_grid} shows repeated generations from a fixed conditioning measurement \(\bx_k\), together with the sample mean and standard deviation maps. At small \(N_k\), the generated samples vary substantially, and the standard deviation is concentrated around QR-code structures. As \(N_k\) increases, the samples become more consistent, reflecting the transition from a stochastic to a deterministic regime for QR reconstruction.

\begin{figure*}[t]
    \centering
    \includegraphics[width=\textwidth]{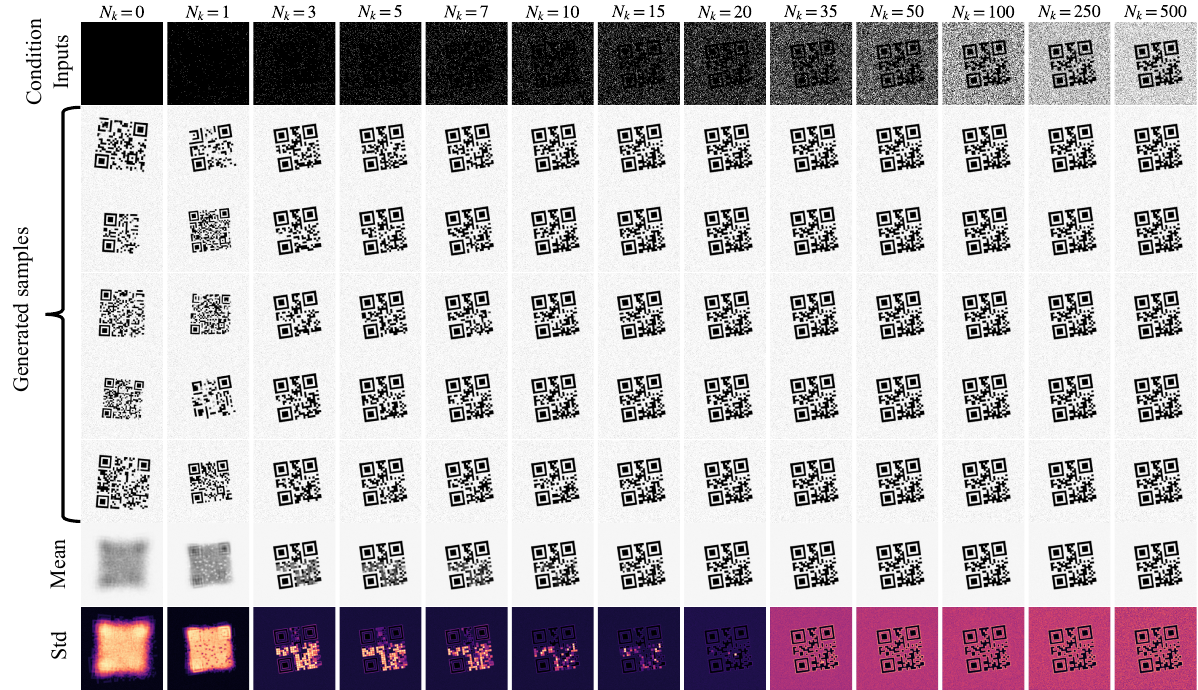}
    \caption{Conditional QR-code generation results from the multistep model. Each column fixes the conditioning measurement \(\bx_k\) at a given frame count \(N_k\), followed by generated samples, their sample mean, and their sample standard deviation. Each standard deviation map is normalized to have maximum value of 1.}
    \label{fig:qr_gen_sample_grid}
\end{figure*}

\paragraph{Comparison with raw SPAD measurements.}
% Since we exploited the off-the-shelf QR code decoder purposed for any generic QR code images, we perform a comparison of decoding accuracy of the generated samples with the raw SPAD measurements without any restoration. 
Finally, we compare decoding accuracy between restored QR-code samples and raw SPAD measurements without restoration. This comparison uses the same off-the-shelf QR-code decoder for both inputs. Figure~\ref{fig:qr_decode_accuracy_pairwise_vs_spad_broken_axis} shows that raw SPAD measurements require approximately \(N_k=5000\) frames to reach stable decoding accuracy, whereas generated samples from the multistep model reach stable decoding with roughly \(N_k=100\) frames. This result shows that the proposed generative restoration model enables task-level determinism at a much lower photon budget than raw SPAD measurements alone.

\begin{figure}[t]
    \centering
    \includegraphics[width=0.9\columnwidth]{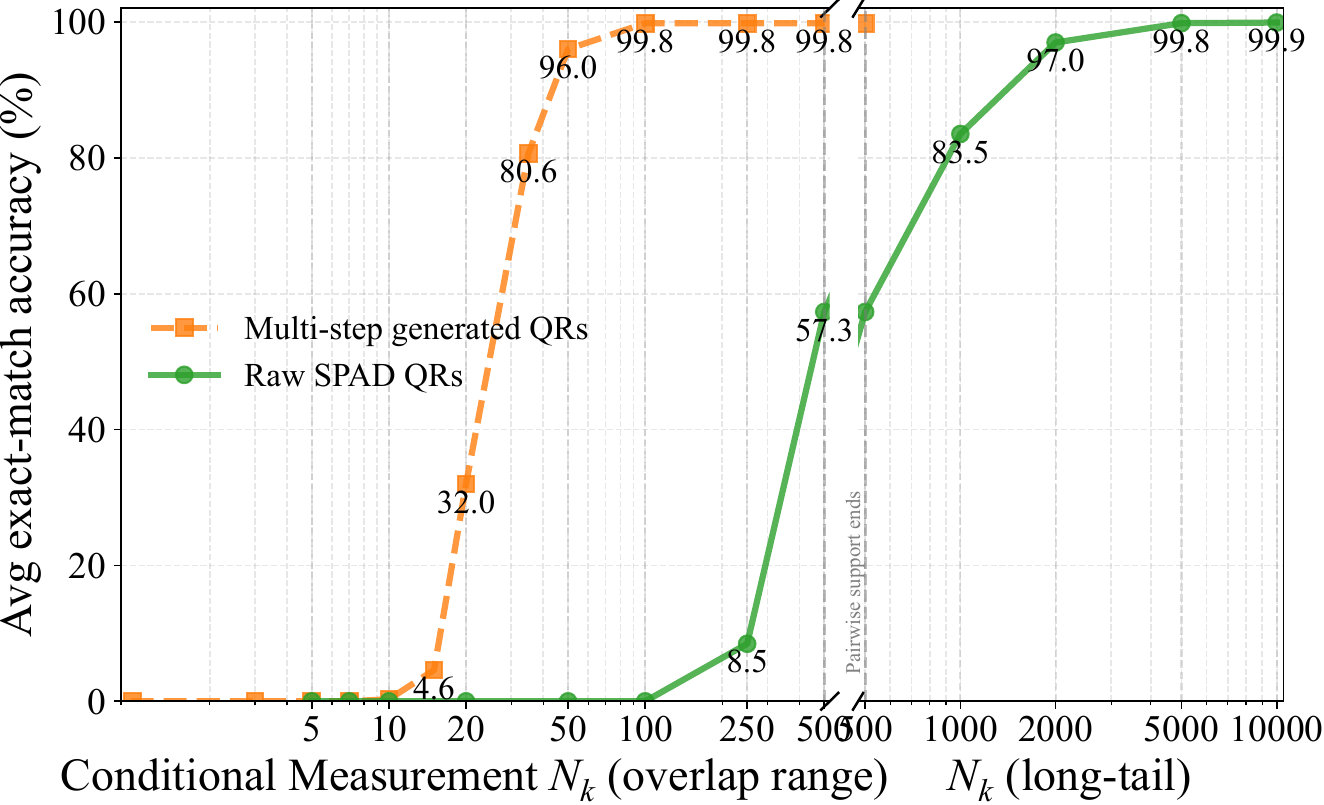}
    \caption{QR-code decoding accuracy of multistep generated samples compared with raw SPAD measurements. The generative model reaches stable decoding accuracy with far fewer conditioning frames.}
    \label{fig:qr_decode_accuracy_pairwise_vs_spad_broken_axis}
\end{figure}

%===============================================
\section{Case Study: Faces}
\label{sec:faces}
% Performing the last case study on natural image generation, we chose face images from FFHQ dataset~\cite{karras2019style}. We resized the images to \(64\times64\) pixels for the generation task. We then trained the proposed conditional SPAD generator, and evaluated the generation quality by the landmark detection and pose estimation performance.

As a final case study, we analyze the results of the proposed method on natural face images. We use grayscale \(64\times64\) images derived from the FFHQ dataset~\cite{karras2019style}, train the proposed conditional SPAD generator, and evaluate whether generated samples preserve facial geometry through landmark detection and pose estimation.

\subsection{Experimental settings}
\paragraph{FFHQ-64 grayscale dataset.}
We use a  pre-processed grayscale version of FFHQ derived from the raw PNG images, which contains 70,000 face images.
We follow the train/validation split recommended by the official FFHQ release:
images with index <60000 are used for training and images with index >=60000
are used for validation, yielding 60,000 training images and 10,000 validation
images.
Each raw face image is converted to grayscale and resized to \(64\times64\) pixels.

\paragraph{SPAD simulation.}
Compared with the previous case studies, face images contain richer structure across multiple scales, from global pose and face shape to fine details such as eyes, mouth, and hair. We therefore use a higher light level, \(\tau=0.1\), and a larger high-count target of \(N_1=100{,}000\) frames. During training, conditioning measurements are sampled with \(N_k\in[0,10{,}000]\). We use intermediate anchor counts
\[
\mathcal A = \{N_1, 10{,}000, 1{,}000, 500, 100, 20, 10, 5 \}.
\]

\subsection{Results}
% \paragraph{Qualitative generation.}
% Figure~\ref{fig:ffhq_gen_sample_grid} shows conditional generation results for face images. As expected, the generated samples are highly stochastic at small \(N_k\), but become increasingly consistent as the number of conditioning frames grows. Noticeably, we can observe that different facial features become stable at different photon budgets. In the shown examples, global structure such as face shape and pose becomes consistent with only a few frames, while finer details such as eyes, mouth, and hair require larger \(N_k\). In the second example, the presence of glasses is already consistent from a single binary SPAD measurement, and the sunglasses' appearance becomes stable by \(N_k=3\).

\paragraph{Downstream task: facial geometry prediction.}
To evaluate task-level determinism for faces, we measure whether generated samples preserve facial geometry. We apply the off-the-shelf MediaPipe Face Landmarker~\cite{lugaresi2019mediapipe,google_mediapipe_facelandmarker} to both generated samples and high-count target images with \(N_1=100{,}000\) frames. We then compare detected facial landmarks and estimated head pose.

Figure~\ref{fig:ffhq64_landmarks_metrics} reports the median landmark and pose errors, with the 25th--75th percentile range. As \(N_k\) increases, both errors decrease and their distributions become tighter, indicating that the generated samples become more deterministic in facial geometry.

Figure~\ref{fig:ffhq64_mediapipe_cond_generated_landmarks} visualizes the detected landmarks on generated samples and raw SPAD measurements. MediaPipe detects landmarks on all generated samples, but fails on raw SPAD measurements until \(N_k=10{,}000\) frames. In contrast, landmarks and pose on the generated samples are already mostly stable around \(N_k=20\). This shows that the generative model produces task-usable facial structure at much lower photon budgets than raw SPAD measurements.

\begin{figure}[t]
    \centering
    \includegraphics[width=1\columnwidth]{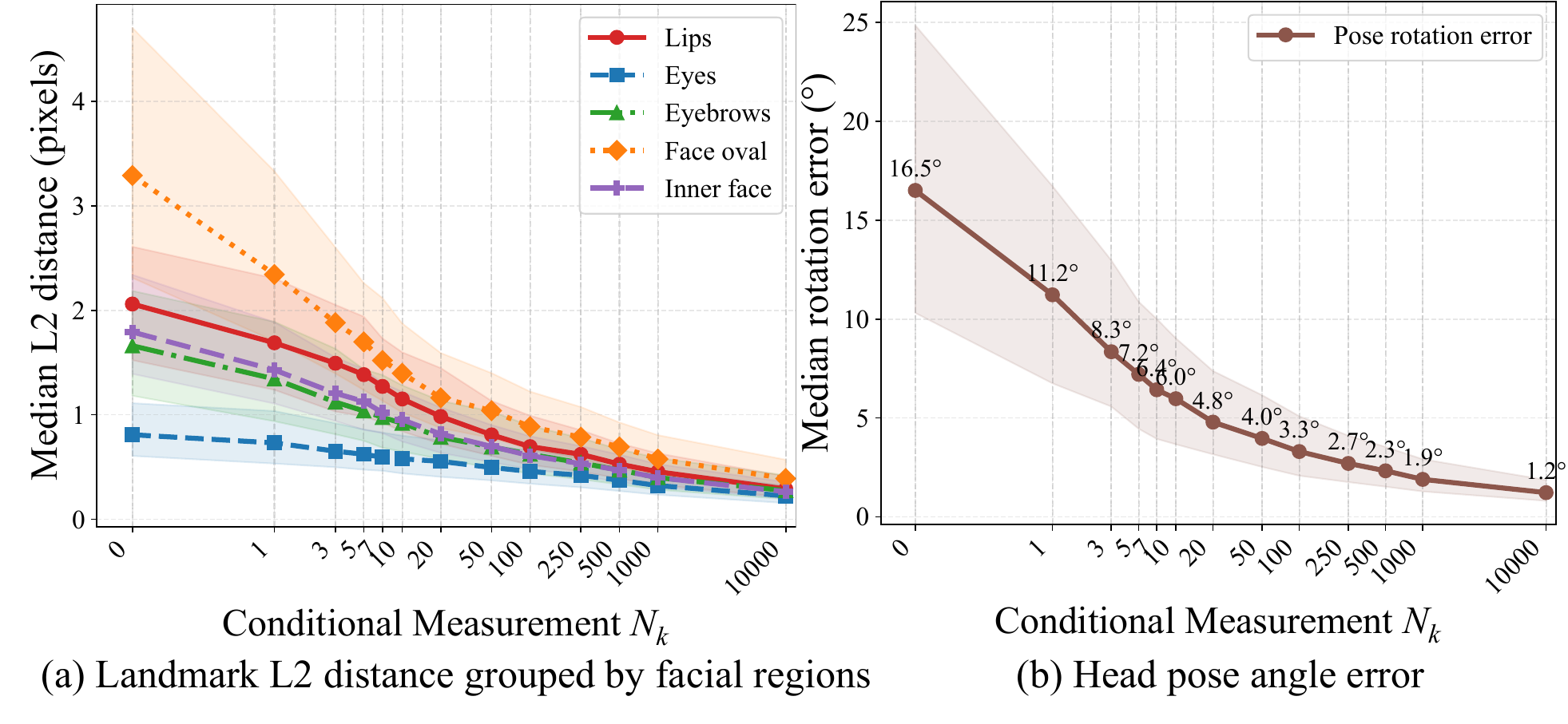}
    \caption{Facial-geometry errors for generated FFHQ samples. We compare MediaPipe landmarks and head pose estimated from generated samples against those estimated from high-count target images with \(N_1=100{,}000\) frames.}
    \label{fig:ffhq64_landmarks_metrics}
\end{figure}

\begin{figure}[t]
    \centering
    \includegraphics[width=1\columnwidth]{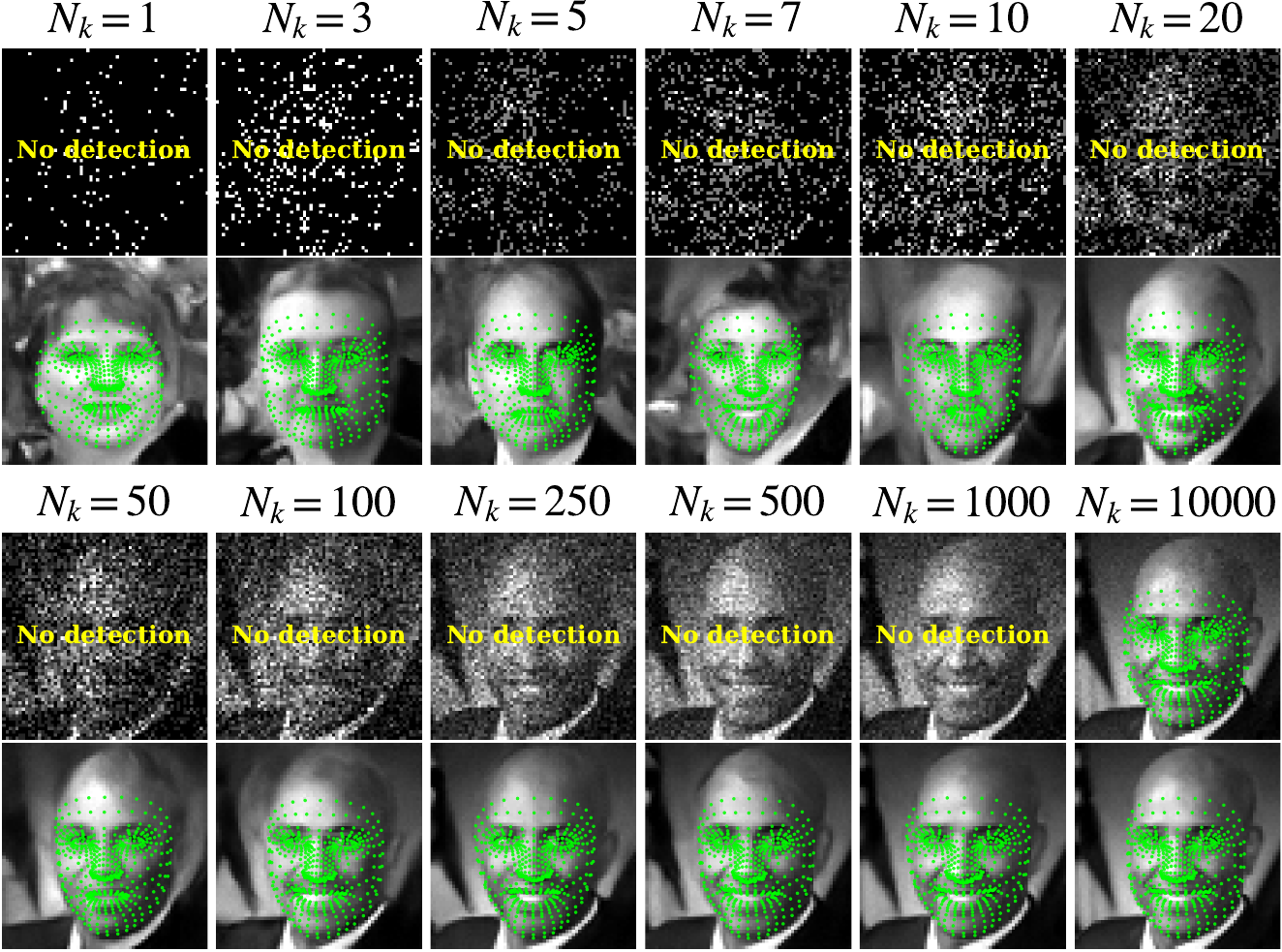}
    \caption{MediaPipe landmark visualizations on generated FFHQ samples and raw SPAD measurements. The generated samples support landmark detection at much lower frame counts than the raw SPAD inputs.}
    \label{fig:ffhq64_mediapipe_cond_generated_landmarks}
\end{figure}

\paragraph{Results on real data.}
Finally, we evaluate the FFHQ model on real SPAD measurements. We captured face images with the SPAD camera and performed the same preprocessing as the synthetic data: suppressing hot pixels and downsizing to \(64\times64\) pixels.

Figure~\ref{fig:ffhq_gen_sample_grid} shows conditional generation results for the real captured face image. As expected, the generated samples are highly stochastic at small \(N_k\), but become increasingly consistent as the number of conditioning frames grows. Noticeably, we can observe that different facial features become stable at different photon budgets. In the shown examples, global structure such as face shape and pose becomes consistent with only a few frames, while finer details such as eyes, mouth, and hair require larger \(N_k\). In the second example, the presence of glasses is already consistent from a single binary SPAD measurement, and the sunglasses' appearance becomes stable by \(N_k=3\).

\begin{figure*}[t]
    \centering
    \includegraphics[width=1\textwidth]{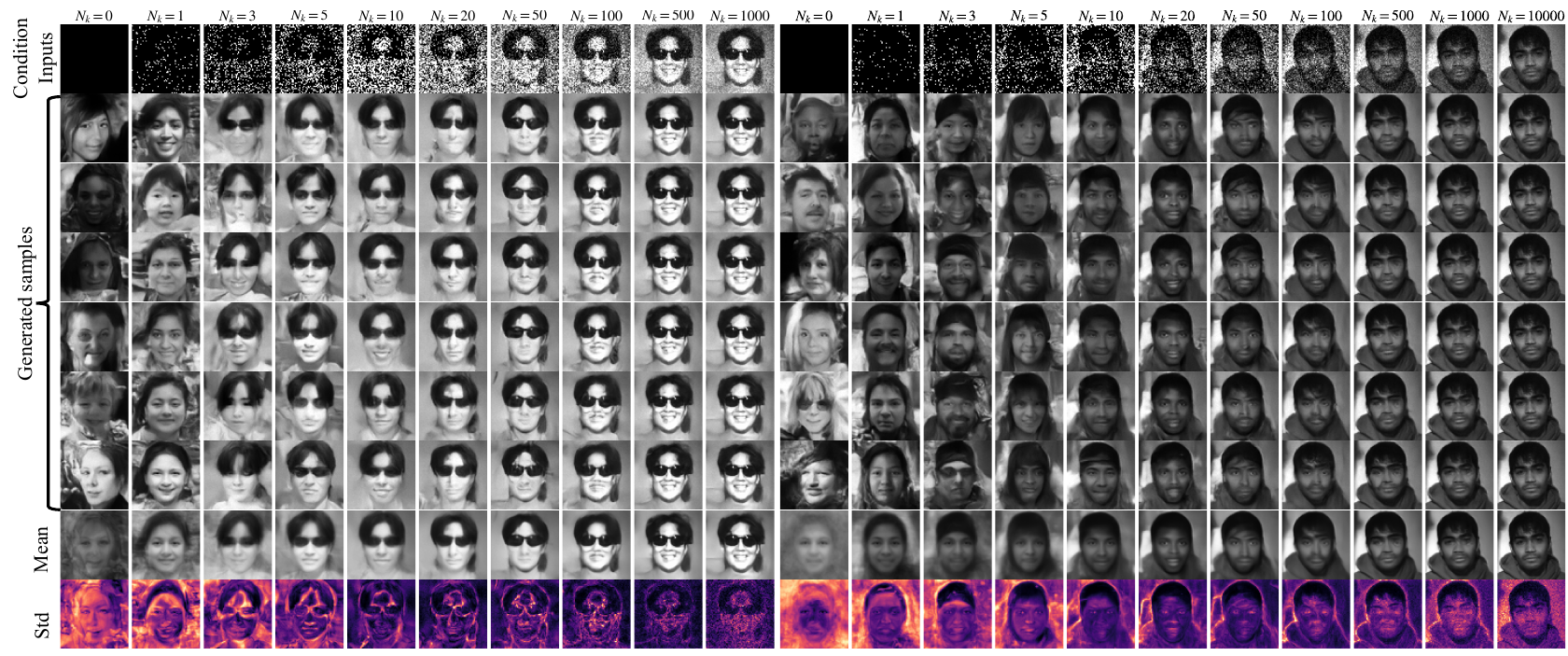}
    \caption{Conditional generation results for real captured face images. For each fixed conditioning measurement \(\bx_k\), the figure shows generated high-count samples, their sample mean, and their sample standard deviation.}
    \label{fig:ffhq_gen_sample_grid}
\end{figure*}

% \section{Case Study: Natural Images}
% \label{sec:natural}

%% file: discuss.tex
\paragraph{A phase transition.} Across multiple tasks, we see that there is a phase transition, typically between 20 and 50 frames, where the performance of classifiers or regressors is very close to its saturated value (i.e., the value at very large photon counts). 
That this phase transition occurs for such low photon counts is an important conclusion of this work.
This has two immediate consequences.
First, it suggests that operating SPAD sensors between four and six bits is a good operating point as far as inference performance is concerned; such higher bit depth readout provides significant reductions in the readout bandwidth.
Second, if we are interested in reliable downstream inference, it suggests that we are likely to incur a 20$\times$  to 50$\times$ penalty in the time resolution of the sensor; given that commercially available sensors today operate at rates as high as 100 kHz, our work suggests that viable inference is possible in 2-5 kHz range.

\paragraph{Limitations.} A key limitation of this work is that its conclusions are all empirical, and hence, subject to the specifics of the training dataset and vagaries of the protocols adopted.
A theoretical counterpart to the phase transition results would help shed deeper insights into the nature of the conditional posterior.
However, such theory has remained elusive in other instances.

Another limitation is the intensive compute needed for successful training of the proposed conditional generative model.
We reduced the resolution of images, especially for the faces dataset, to keep training times tractable/
For harder problems, such as natural images, we would also need to use more anchor frames, which would further acerbate the compute concerns.

\paragraph{The high flux regime.} In high flux regimes, the probability of multiple photons arriving within each exposure becomes non-negligible. 
Surprisingly, this regime also results in a one-to-many problem since doubling or tripling the photon arrival rate doesn't results in a scaling of the SPAD measurements.
The techniques developed in this paper are, in principle, capable of handling this scenario since the Binomial marginals applies just as well to the high flux case; a cavaet here is that we would need to train diffusion models capable of generating high dynamic range images.

%\paragraph{Is this analysis unique to SPAD sensors?}
\paragraph{Conclusion.} This paper provides tools to model image synthesis over a photon budget landscape. 
We argue that in photon starved regimes the reconstruction problem is a one-to-many problem, a feature that has not been explored in prior work.
To realize this, we develop a consistent theory of diffusion models with pixelwise Binomial distribution; this provides a toolbox for studying the one-to-many nature of the inverse problem, and how reconstruction and inference performance changes as the inverse problem becomes solvable.
We believe this provide invaluable insights in the nature of SPAD sensors, in terms of practical operating points for them.

%% file: appendix.tex
\appendix

\section{Frame-Thinning Markov Chain}
\label{app:frame-thinning-proof}

We prove Theorem~\ref{thm:frame-thinning-markov}.  Let
\[
K_t(\bx_{t+1}\mid \bx_t)
:=
\prod_{i=1}^D
\operatorname{Hypergeometric}(x_{t+1,i};N_t,x_{t,i},N_{t+1}),
\]
where the first argument is the kept count, the population size is $N_t$, the number of successes is $x_{t,i}$, and the number of draws is $N_{t+1}$.

\begin{proof}
Define the forward path distribution by
\begin{equation}
q(\bx_{1:L}\mid \brho)
=
q(\bx_1\mid \brho)
\prod_{t=1}^{L-1}K_t(\bx_{t+1}\mid \bx_t).
\label{eq:app-forward-path}
\end{equation}
Then
\[
q(\bx_{t+1}\mid \bx_{1:t},\brho)=K_t(\bx_{t+1}\mid \bx_t),
\]
so the process is Markov by construction.

It remains to show that the Hypergeometric kernel preserves the Binomial marginal. It suffices to prove the scalar case. Suppose
\[
X_t\mid \rho \sim \operatorname{Binomial}(N_t,\rho),
\]
and
\[
X_{t+1}\mid X_t=k
\sim
\operatorname{Hypergeometric}(N_t,k,N_{t+1}).
\]
For $m\in\{0,\ldots,N_{t+1}\}$,
\begin{align}
q(X_{t+1}=m\mid \rho)
={}&
\sum_k
\binom{N_t}{k}\rho^k(1-\rho)^{N_t-k}
\notag \\
&\times
\frac{\binom{k}{m}\binom{N_t-k}{N_{t+1}-m}}{\binom{N_t}{N_{t+1}}}.
\end{align}
Using
\[
\binom{N_t}{k}\binom{k}{m}
=
\binom{N_t}{m}\binom{N_t-m}{k-m},
\]
and setting $y=k-m$, we get
\begin{align}
q(X_{t+1}=m\mid \rho)
={}&
\frac{\binom{N_t}{m}}{\binom{N_t}{N_{t+1}}}
\rho^m(1-\rho)^{N_{t+1}-m}
\notag \\
&\times
\sum_{y=0}^{N_t-N_{t+1}}
\binom{N_t-m}{y}
\binom{N_t-m-y}{N_{t+1}-m}
\notag \\
&\quad\times
\rho^y(1-\rho)^{N_t-N_{t+1}-y}.
\end{align}
The identity
\[
\binom{A}{y}\binom{A-y}{B}
=
\binom{A}{B}\binom{A-B}{y}
\]
with $A=N_t-m$ and $B=N_{t+1}-m$ gives
\begin{align}
q(X_{t+1}=m\mid \rho)
={}&
\frac{\binom{N_t}{m}\binom{N_t-m}{N_{t+1}-m}}{\binom{N_t}{N_{t+1}}}
\rho^m(1-\rho)^{N_{t+1}-m}
\notag \\
&\times
\sum_{y=0}^{N_t-N_{t+1}}
\binom{N_t-N_{t+1}}{y}
\rho^y(1-\rho)^{N_t-N_{t+1}-y}.
\end{align}
The sum equals $1$, and
\[
\frac{\binom{N_t}{m}\binom{N_t-m}{N_{t+1}-m}}{\binom{N_t}{N_{t+1}}}
=
\binom{N_{t+1}}{m}.
\]
Therefore
\[
X_{t+1}\mid \rho
\sim
\operatorname{Binomial}(N_{t+1},\rho).
\]
Applying this independently over pixels proves
\[
\bX_{t+1}\mid \brho
\sim
\prod_i\operatorname{Binomial}(N_{t+1},\rho_i).
\]
\end{proof}

\section{Reverse Frame-Addition Kernel}
\label{app:true-reverse-proof}

Let $M_t=N_t-N_{t+1}$ and $\Delta\bx_t=\bx_t-\bx_{t+1}$. Conditioned on $\brho$, the missing $M_t$ binary frames are independent of the already retained $N_{t+1}$ frames. Hence
\begin{equation}
\bX_t=\bX_{t+1}+\bY_t,
\qquad
Y_{t,i}\mid \brho\sim \operatorname{Binomial}(M_t,\rho_i)
\label{eq:app-oracle-frame-addition}
\end{equation}
independently over $i$. Thus the oracle reverse kernel is
\begin{equation}
q(\bx_t\mid \bx_{t+1},\brho,N_t,N_{t+1})
=
\prod_{i=1}^D
\operatorname{Binomial}(x_{t,i}-x_{t+1,i};M_t,\rho_i),
\label{eq:app-oracle-reverse-kernel}
\end{equation}
with support $x_{t,i}\ge x_{t+1,i}$ and $x_{t,i}-x_{t+1,i}\le M_t$.

If $\nu_{\mathrm{true}}(\brho\mid \bx_{t+1},N_{t+1})$ denotes the posterior distribution of $\brho$ after observing $\bX_{t+1}=\bx_{t+1}$, then marginalizing $\brho$ gives
\begin{equation}
\begin{aligned}
&q(\bx_t\mid \bx_{t+1},N_t,N_{t+1}) \\
&\quad=
\int
\left[
\prod_{i=1}^D
\operatorname{Binomial}(x_{t,i}-x_{t+1,i};M_t,\rho_i)
\right] \\
&\qquad\times
\nu_{\mathrm{true}}(d\brho\mid \bx_{t+1},N_{t+1}).
\end{aligned}
\label{eq:app-true-reverse-mixture}
\end{equation}

\section{ELBO Derivation}
\label{app:elbo-proof}

Let
\[
c_t=(\bx_{t+1},N_t,N_{t+1})
\]
denote the condition. The latent marginal likelihood is
\[
p_{\theta,\psi}(\bx_t\mid c_t)
=
\int p_\theta(\bx_t\mid c_t,\bz)p_\psi(\bz\mid c_t)d\bz.
\]
For any variational posterior $q_\phi(\bz\mid \bx_t,c_t)$,
\begin{align}
&\log p_{\theta,\psi}(\bx_t\mid c_t)
\notag \\
&=
\log \int
q_\phi(\bz\mid \bx_t,c_t)
\frac{p_\theta(\bx_t\mid c_t,\bz)p_\psi(\bz\mid c_t)}{q_\phi(\bz\mid \bx_t,c_t)}d\bz
\notag \\
&\ge
\mathbb E_{q_\phi}
\big[
\log p_\theta(\bx_t\mid c_t,\bz)
+\log p_\psi(\bz\mid c_t)
\notag \\
&\qquad\quad
-\log q_\phi(\bz\mid \bx_t,c_t)
\big]
\notag \\
&=
\mathbb E_{q_\phi}
\!\left[
\log p_\theta(\bx_t\mid c_t,\bz)
\right]
\notag \\
&\quad
-
\operatorname{KL}
\!\left(q_\phi(\bz\mid \bx_t,c_t)\,\|\,p_\psi(\bz\mid c_t)\right).
\end{align}
The inequality is Jensen's inequality.

\section{Loss Terms}
\label{app:loss-derivations}

\subsection{Binomial logit loss}

Given $\bz$, let $\brho_\theta=\sigma(\bu)$ and $M_t=N_t-N_{t+1}$. Then
\begin{align}
&p_\theta(\bx_t\mid \bx_{t+1},N_t,N_{t+1},\bz)
\notag \\
&\quad=
\prod_{i=1}^D
\binom{M_t}{\Delta x_{t,i}}
\rho_{\theta,i}^{\Delta x_{t,i}}
(1-\rho_{\theta,i})^{M_t-\Delta x_{t,i}}.
\end{align}
Using
\[
\log\sigma(u)=u-\operatorname{softplus}(u),
\qquad
\log(1-\sigma(u))=-\operatorname{softplus}(u),
\]
we get
\begin{align}
&-\log p_\theta(\bx_t\mid \bx_{t+1},N_t,N_{t+1},\bz)
\notag \\
&\quad=
-\sum_i\log\binom{M_t}{\Delta x_{t,i}}
\notag \\
&\qquad+
\sum_i
\left(M_t\operatorname{softplus}(u_i)-\Delta x_{t,i}u_i\right).
\end{align}
The first term is independent of the model parameters and is dropped in the practical loss.

\subsection{Reparameterization}

For a diagonal Gaussian posterior
\[
q_\phi(\bz\mid \bx_t,\bx_{t+1},N_t,N_{t+1})
=
\mathcal N(\bm{z}; \bmu^q,\operatorname{diag}(e^{\bell^q})),
\]
a sample is written as
\[
\bz=\bmu^q+\exp\left(\tfrac12\bell^q\right)\odot\beps,
\qquad
\beps\sim\mathcal N(\bm{0},\bm{I}).
\]
The factor $1/2$ converts log-variance to standard deviation.

\subsection{Gaussian KL}

Let
\[
q=\mathcal N(\bmu^q,\operatorname{diag}(e^{\bell^q})),
\qquad
p=\mathcal N(\bmu^p,\operatorname{diag}(e^{\bell^p})).
\]
Then
\begin{equation}
\operatorname{KL}(q\|p)
=
\frac12\sum_r
\left[
\ell_r^p-\ell_r^q-1
+e^{\ell_r^q-\ell_r^p}
+(\mu_r^q-\mu_r^p)^2e^{-\ell_r^p}
\right].
\label{eq:app-gaussian-kl}
\end{equation}
To see this, start from
\[
\operatorname{KL}(q\|p)=\mathbb E_q[\log q(\bz)-\log p(\bz)].
\]
For coordinate $r$,
\[
\mathbb E_q[(Z_r-\mu_r^q)^2]=e^{\ell_r^q},
\qquad
\mathbb E_q[(Z_r-\mu_r^p)^2]=e^{\ell_r^q}+(\mu_r^q-\mu_r^p)^2.
\]
Substitution into the Gaussian log densities gives~\eqref{eq:app-gaussian-kl}.

\subsection{Hierarchical latent groups}
\label{app:hier-latents}

If the latent variable is grouped as
\[
\bz=(\bz_g,\bz_{s_1},\bz_{s_2},\ldots),
\]
the prior and posterior are parameterized by one diagonal Gaussian per group. The total KL used in training is
\begin{equation}
\mathcal J_{\mathrm{KL}}
=
\sum_{h\in\mathcal G}\mathcal J_{\mathrm{KL},h},
\label{eq:app-group-kl}
\end{equation}
where each groupwise term has the form~\eqref{eq:app-gaussian-kl}, summed over all channels and spatial positions of that group.

\subsection{Gap-normalized reconstruction loss}

For one pixel, let $\hat\rho_i=\sigma(u_i)$ and assume the oracle increment satisfies
\[
\Delta X_{t,i}\mid \rho_i\sim\operatorname{Binomial}(M_t,\rho_i).
\]
The expected raw logit loss is
\begin{align}
\mathbb E[M_t\operatorname{softplus}(u_i)-\Delta X_{t,i}u_i\mid \rho_i]
&=
M_t\left(\operatorname{softplus}(u_i)-\rho_i u_i\right).
\end{align}
The excess risk above the optimum $\hat\rho_i=\rho_i$ is
\begin{equation}
M_t\operatorname{KL}
\left(
\operatorname{Bernoulli}(\rho_i)
\| 
\operatorname{Bernoulli}(\hat\rho_i)
\right).
\label{eq:app-raw-excess-risk}
\end{equation}
Thus the raw reconstruction term scales linearly with the frame gap $M_t$. Dividing the reconstruction loss by $M_t$ removes this explicit scale factor. This is a training reweighting only; it does not change the generative kernel.

\subsection{Single-step vs. multistep generation.}
We compare single-step and multistep reverse generation using QR decoding accuracy. Because the proposed pairwise model is conditioned on the source and target frame counts, it can be trained either to jump directly from \(\bx_k\) to \(\bx_1\), or to reverse the SPAD ladder through intermediate anchor frames. The single-step model directly estimates
\[
p_{\theta,\psi}(\bx_1 \mid \bx_k, N_1, N_k),
\]
whereas the multistep model uses the anchor-based reverse chain described above.

For evaluation, we use the OpenCV \texttt{WeChatQRCode} decoder. We select 1,000 QR-code images from the validation set, simulate SPAD measurements,
generate restored samples conditioned on each measurement, and compute the decoding accuracy. Figure~\ref{fig:qr_decode_accuracy_vs_Nk_multi_vs_single} shows the results. The multistep model consistently achieves higher decoding accuracy than the single-step model in the stochastic regime, before performance saturates.

Figure~\ref{fig:qr_single_vs_multi_examples} shows representative generations from the same low-count conditioning measurements. The multistep model produces sharper and more QR-like structures than the single-step model. These results support the use of intermediate anchors: decomposing a difficult restoration problem into smaller reverse transitions improves generation quality, consistent with the motivation behind multistep generative models such as diffusion and flow models.

\begin{figure}[H]
    \centering
    \includegraphics[width=0.9\columnwidth]{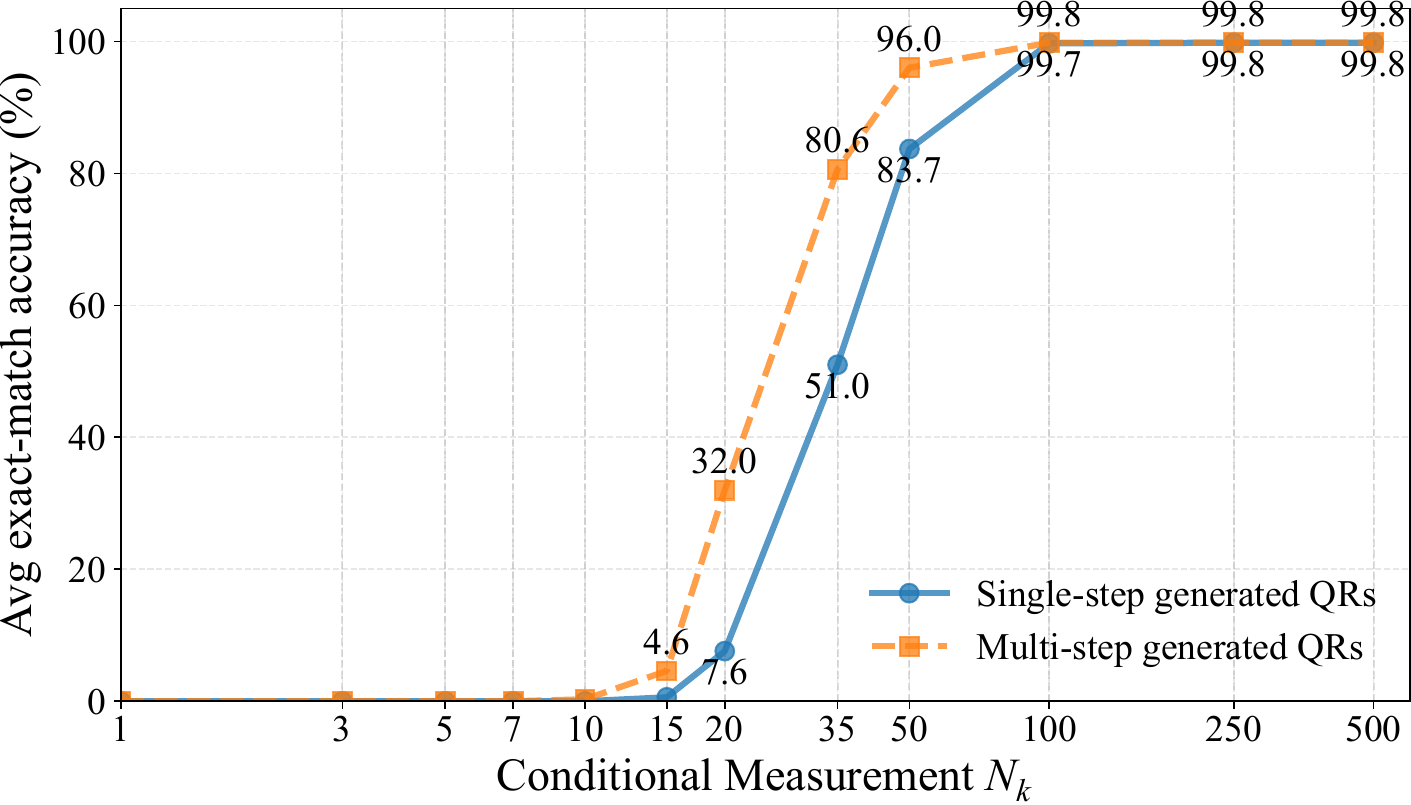}
    \caption{QR-code decoding accuracy of generated samples as a function of the number of conditioning SPAD frames \(N_k\).}
    \label{fig:qr_decode_accuracy_vs_Nk_multi_vs_single}
\end{figure}

\begin{figure}[H]
    \centering
    \includegraphics[width=1\columnwidth]{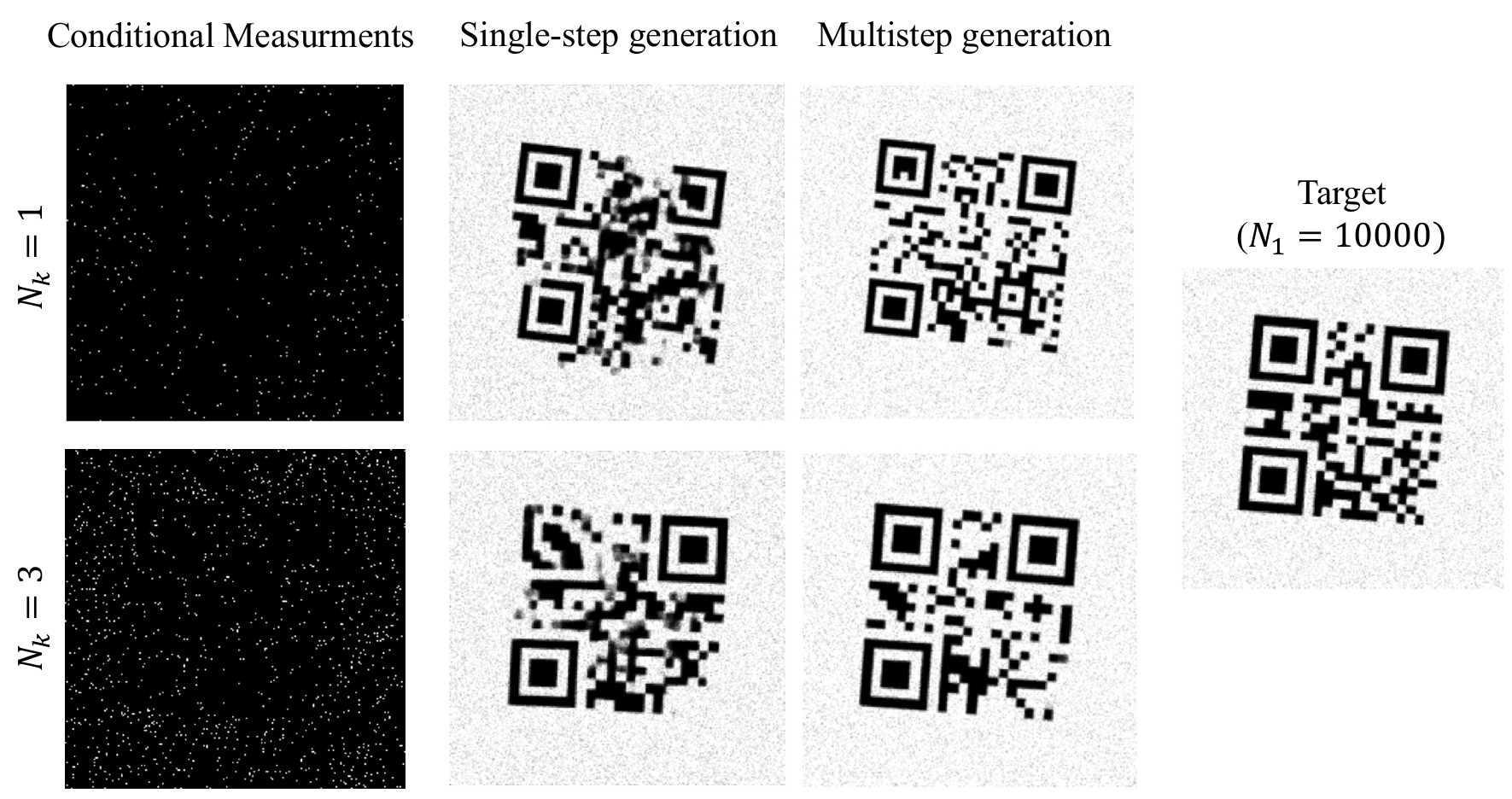}
    \caption{Comparison of single-step and multistep generated QR-code samples from the same low-count conditioning measurements.}
    \label{fig:qr_single_vs_multi_examples}
\end{figure}

%% file: refs.bib
@inproceedings{sohl2015deep,
  title={Deep unsupervised learning using nonequilibrium thermodynamics},
  author={Sohl-Dickstein, Jascha and Weiss, Eric and Maheswaranathan, Niru and Ganguli, Surya},
  booktitle={International conference on machine learning},
  pages={2256--2265},
  year={2015},
}

@article{ho2020denoising,
  title={Denoising diffusion probabilistic models},
  author={Ho, Jonathan and Jain, Ajay and Abbeel, Pieter},
  journal={Advances in neural information processing systems},
  volume={33},
  pages={6840--6851},
  year={2020}
}

@article{austin2021structured,
  title={Structured denoising diffusion models in discrete state-spaces},
  author={Austin, Jacob and Johnson, Daniel D and Ho, Jonathan and Tarlow, Daniel and Van Den Berg, Rianne},
  journal={Advances in neural information processing systems},
  volume={34},
  pages={17981--17993},
  year={2021}
}

@article{sohn2015learning,
  title={Learning structured output representation using deep conditional generative models},
  author={Sohn, Kihyuk and Lee, Honglak and Yan, Xinchen},
  journal={Advances in neural information processing systems},
  volume={28},
  year={2015}
}

@article{kingma2013auto,
  title={Auto-encoding variational bayes},
  author={Kingma, Diederik P and Welling, Max},
  journal={arXiv preprint arXiv:1312.6114},
  year={2013}
}

@article{sonderby2016ladder,
  title={Ladder variational autoencoders},
  author={S{\o}nderby, Casper Kaae and Raiko, Tapani and Maal{\o}e, Lars and S{\o}nderby, S{\o}ren Kaae and Winther, Ole},
  journal={Advances in neural information processing systems},
  volume={29},
  year={2016}
}

@article{vahdat2020nvae,
  title={{NVAE: A} deep hierarchical variational autoencoder},
  author={Vahdat, Arash and Kautz, Jan},
  journal={Advances in neural information processing systems},
  volume={33},
  pages={19667--19679},
  year={2020}
}

@article{child2020very,
  title={Very deep vaes generalize autoregressive models and can outperform them on images},
  author={Child, Rewon},
  journal={arXiv preprint arXiv:2011.10650},
  year={2020}
}

@inproceedings{sonderby2016train,
  title={How to train deep variational autoencoders and probabilistic ladder networks},
  author={S{\o}nderby, Casper Kaae and Raiko, Tapani and Maal{\o}e, Lars and S{\o}nderby, S{\o}ren Kaae and Winther, Ole},
  booktitle={33rd International Conference on Machine Learning (ICML 2016)},
  year={2016}
}

@article{sancar2025,
author = "yasin sancar",
title = "{QR Code Dataset V2}",
year = "2025",
month = "2",
url = "https://figshare.com/articles/dataset/QR_Code_Dataset_V2/28424213",
doi = "10.6084/m9.figshare.28424213.v1"
}

@article{chen2019adaptive,
  title={Adaptive binarization of QR code images for fast automatic sorting in warehouse systems},
  author={Chen, Rongjun and Yu, Yongxing and Xu, Xiansheng and Wang, Leijun and Zhao, Huimin and Tan, Hong-Zhou},
  journal={Sensors},
  volume={19},
  number={24},
  pages={5466},
  year={2019},
  publisher={MDPI}
}

@article{qian2017optimization,
  title={Optimization of QR code readability in movement state using response surface methodology for implementing continuous chain traceability},
  author={Qian, Jianping and Du, Xiaowei and Zhang, Baoyan and Fan, Beilei and Yang, Xinting},
  journal={Computers and Electronics in Agriculture},
  volume={139},
  pages={56--64},
  year={2017},
  publisher={Elsevier}
}

@inproceedings{karras2019style,
  title={A style-based generator architecture for generative adversarial networks},
  author={Karras, Tero and Laine, Samuli and Aila, Timo},
  booktitle={Proceedings of the IEEE/CVF conference on computer vision and pattern recognition},
  pages={4401--4410},
  year={2019}
}

@inproceedings{lugaresi2019mediapipe,
  title     = {{MediaPipe}: A Framework for Perceiving and Processing Reality},
  author    = {Lugaresi, Camillo and Tang, Jiuqiang and Nash, Hadon and McClanahan, Chris and Uboweja, Esha and Hays, Michael and Zhang, Fan and Chang, Chuo-Ling and Yong, Ming Guang and Lee, Juhyun and Chang, Wan-Teh and Hua, Wei and Georg, Manfred and Grundmann, Matthias},
  booktitle = {Third Workshop on Computer Vision for AR/VR at IEEE Computer Vision and Pattern Recognition},
  year      = {2019}
}

@misc{google_mediapipe_facelandmarker,
  title        = {{MediaPipe} Face Landmarker},
  author       = {{Google AI Edge}},
  howpublished = {\url{https://ai.google.dev/edge/mediapipe/solutions/vision/face_landmarker}},
  note         = {Accessed: 2026-05-10}
}

@article{lecun2002gradient,
  title={Gradient-based learning applied to document recognition},
  author={LeCun, Yann and Bottou, L{\'e}on and Bengio, Yoshua and Haffner, Patrick},
  journal={Proceedings of the IEEE},
  volume={86},
  number={11},
  pages={2278--2324},
  year={2002},
  publisher={Ieee}
}

@inproceedings{zhu2023denoising,
  title={Denoising diffusion models for plug-and-play image restoration},
  author={Zhu, Yuanzhi and Zhang, Kai and Liang, Jingyun and Cao, Jiezhang and Wen, Bihan and Timofte, Radu and Van Gool, Luc},
  booktitle={Proceedings of the IEEE/CVF conference on computer vision and pattern recognition},
  pages={1219--1229},
  year={2023}
}

@inproceedings{wei2023passive,
  title={Passive ultra-wideband single-photon imaging},
  author={Wei, Mian and Nousias, Sotiris and Gulve, Rahul and Lindell, David B and Kutulakos, Kiriakos N},
  booktitle={Proceedings of the IEEE/CVF International Conference on Computer Vision},
  pages={8135--8146},
  year={2023}
}

@inproceedings{liu2025noise,
  title={Noise-Modeled Diffusion Models for Low-Light Spike Image Restoration},
  author={Liu, Ruonan and Zhu, Lin and Xiang, Xijie and Wang, Lizhi and Huang, Hua},
  booktitle={Proceedings of the IEEE/CVF International Conference on Computer Vision},
  pages={4080--4089},
  year={2025}
}

@inproceedings{chung2023diffusion,
  title={DIFFUSION POSTERIOR SAMPLING FOR GENERAL NOISY INVERSE PROBLEMS},
  author={Chung, Hyungjin and Kim, Jeongsol and McCann, Michael T and Klasky, Marc L and Ye, Jong Chul},
  booktitle={International Conference on Learning Representations},
  year={2023}
}

@article{song2020score,
  title={Score-based generative modeling through stochastic differential equations},
  author={Song, Yang and Sohl-Dickstein, Jascha and Kingma, Diederik P and Kumar, Abhishek and Ermon, Stefano and Poole, Ben},
  journal={arXiv preprint arXiv:2011.13456},
  year={2020}
}

@article{melidonis2025score,
  title={Score-based denoising diffusion models for photon-starved image restoration problems},
  author={Melidonis, Savvas and Xi, Yiming and Zygalakis, Konstantinos C and Altmann, Yoann and Pereyra, Marcelo},
  journal={Transactions on Machine Learning Research},
  year={2025}
}

@article{dvir2026diffusion,
  title={Diffusion in SPAD Signals},
  author={Dvir, Lior and Torem, Nadav and Schechner, Yoav Y},
  journal={arXiv preprint arXiv:2601.07599},
  year={2026}
}

@inproceedings{chennuri2025quanta,
  title={Quanta Diffusion},
  author={Chennuri, Prateek and Fu, Dongdong and Chan, Stanley H},
  booktitle={IEEE International Conference on Image Processing (ICIP)},
  pages={229--234},
  year={2025},
}

@inproceedings{chennuri2024quanta,
  title={Quanta video restoration},
  author={Chennuri, Prateek and Chi, Yiheng and Jiang, Enze and Godaliyadda, GM Dilshan and Gnanasambandam, Abhiram and Sheikh, Hamid R and Gyongy, Istvan and Chan, Stanley H},
  booktitle={European Conference on Computer Vision},
  pages={152--171},
  year={2024},
}

@article{lai2025poisson,
  title={Poisson diffusion probabilistic model for low-dose SPECT sinogram denoising},
  author={Lai, Peng and Wu, Ruifan and Yuan, Woliang and Li, Haiying and Jiang, Ying},
  journal={Medical Physics},
  volume={52},
  number={6},
  pages={4250--4265},
  year={2025},
}

@article{nachmani2021non,
  title={Non gaussian denoising diffusion models},
  author={Nachmani, Eliya and Roman, Robin San and Wolf, Lior},
  journal={arXiv preprint arXiv:2106.07582},
  year={2021}
}

@article{xie2023diffusion,
  title={Diffusion model for generative image denoising},
  author={Xie, Yutong and Yuan, Minne and Dong, Bin and Li, Quanzheng},
  journal={arXiv preprint arXiv:2302.02398},
  year={2023}
}

@inproceedings{chen2023learning,
  title={Learning to jump: Thinning and thickening latent counts for generative modeling},
  author={Chen, Tianqi and Zhou, Mingyuan},
  booktitle={International Conference on Machine Learning},
  pages={5367--5382},
  year={2023},
}

@article{dhariwal2021diffusion,
  title={Diffusion models beat gans on image synthesis},
  author={Dhariwal, Prafulla and Nichol, Alexander},
  journal={Advances in neural information processing systems},
  volume={34},
  pages={8780--8794},
  year={2021}
}

@article{goodfellow2014generative,
  title={Generative adversarial nets},
  author={Goodfellow, Ian J and Pouget-Abadie, Jean and Mirza, Mehdi and Xu, Bing and Warde-Farley, David and Ozair, Sherjil and Courville, Aaron and Bengio, Yoshua},
  journal={Advances in neural information processing systems},
  volume={27},
  year={2014}
}

@article{garg2026gqir,
  title={gQIR: Generative Quanta Image Reconstruction},
  author={Garg, Aryan and Ma, Sizhuo and Gupta, Mohit},
  journal={arXiv preprint arXiv:2602.20417},
  year={2026}
}
